\documentclass{article}

\usepackage[preprint]{neurips_2026}
\setcitestyle{numbers,square,comma}

\usepackage[utf8]{inputenc}
\usepackage[T1]{fontenc}
\usepackage{amsmath,amssymb,amsfonts}
\usepackage{amsthm}
\usepackage{algorithm}
\usepackage{algorithmic}
\usepackage{booktabs}
\usepackage{multirow}
\usepackage{subcaption}
\usepackage{xspace}
\usepackage{listings}
\usepackage{tikz}
\usepackage{pgfplots}
\pgfplotsset{compat=1.18}
\usetikzlibrary{shapes.geometric, arrows.meta, positioning, fit, backgrounds, calc, decorations.pathreplacing}
\usepackage[hidelinks]{hyperref}
\usepackage{enumitem}
\usepackage{float}

\newtheorem{theorem}{Theorem}
\newtheorem{lemma}{Lemma}
\newtheorem{corollary}{Corollary}
\newtheorem{proposition}{Proposition}

\newcommand{\odb}{\textsc{ODB}\xspace}

\title{Online Dynamic Batching with Formal Guarantees for LLM Training}

\author{
  Dian Li\textsuperscript{*,\textdagger} \quad
  Zekun Wang\textsuperscript{*} \quad
  Yaoru Wang \quad
  Jiahong Yan \\
  Tencent \\
  \texttt{\{goodli,zekunwang,yaoruwang,redyan\}@tencent.com} \\
  \textsuperscript{*}Equal contribution. \quad
  \textsuperscript{\textdagger}Corresponding author.
}

\begin{document}

\maketitle

\begin{abstract}
Modern LLM training breaks a core assumption behind offline batch
samplers: the true training cost of a sample is only observable after
preprocessing, augmentation, templating, tokenization, and multimodal
visual-token expansion. Unless one pays for a preprocessing- and
augmentation-dependent length cache,
batch construction is therefore blind to the quantity that determines
padding, memory use, and GPU saturation. We introduce
\textbf{Online Dynamic Batching} (\odb), a DataLoader-side drop-in system
that moves batch formation to this point of accurate observability while
preserving DDP step alignment. We formalize this synchronization
requirement as the \textbf{Distributed Group Alignment Problem} and prove
deadlock-free bounded termination with default join-mode identity coverage
and opt-in non-join sample-quota closure. \odb requires no model, optimizer, or
attention-kernel changes
and is released as \texttt{online-dynamic-batching} with lightweight trainer
adapters. Across public 2B/8B Qwen3-VL runs on UltraChat/LLaVA/ShareGPT4o,
\odb improves literal emitted-sample throughput vs.\ fixed-batch Standard by
\textbf{1.58--2.51$\times$} on single-node Full FT/LoRA and
\textbf{1.71--3.78$\times$} on two-node Full FT, with Standard-comparable
quality; production MM-Mix reaches \textbf{4.43$\times$}. Against GMT/BMT
offline token-budget oracles, \odb is within 15\% on UltraChat/LLaVA and faster
on high-CV ShareGPT4o: \textbf{2.24--2.39$\times$} single-node Full FT/LoRA and
\textbf{3.06--3.69$\times$} two-node Full FT. Together, \odb occupies the online/drop-in regime for
high-heterogeneity LLM fine-tuning: large throughput gains at Standard-comparable
quality, formal DGAP guarantees, and no length-cache precompute or kernel rewrites.
\end{abstract}

\section{Introduction}
\label{sec:intro}

\paragraph{Training-time batching has an observability problem.}
In modern LLM and multimodal fine-tuning stacks, the cost of a sample is not a static dataset attribute: it is realized only after preprocessing, augmentation, chat templating, tokenization, and visual-token expansion.
A sampler that runs before this point either ignores the quantity that determines padding and memory use, or pays for a length cache that is tied to a specific transform/template/cutoff policy.
The natural systems move is therefore to form batches exactly where the true length becomes observable.
The difficulty is that runtime variable-size batching breaks a contract that fixed-batch DDP gets for free: all ranks must execute the same number of gradient-reduction steps.
We measure length heterogeneity by $\mathbf{CV}=\sigma/\mu$, the coefficient of variation of post-pipeline tokenized sample lengths.

\paragraph{Three conditions for efficient training.}
We frame any batching strategy as needing to satisfy \emph{three} conditions simultaneously:
(1)~\emph{spatial efficiency}---padding $\approx 0$;
(2)~\emph{compute saturation}---each step's effective workload (tokens$\,\times\,$FLOPs) saturates the GPU;
(3)~\emph{temporal efficiency}---data prep overlaps GPU compute.
Fixed-batch baselines violate at least one condition on variable-length data: small bs avoids padding but underfills the GPU, while large bs raises work per step by mixing unequal lengths and therefore pays padding or OOM cost.
Sequence packing can remove padding, but for multimodal training it is a model/kernel-level intervention rather than a pure DataLoader-level batching method (\S\ref{sec:related}).
Existing HuggingFace length grouping~\cite{wolf2020transformers}, token-budget~\cite{ott2019fairseq}, and pre-bucketed~\cite{kuchaiev2019nemo,mosaicml2022composer} approaches either require offline preprocessing/augmentation assumptions or use fixed batch sizes; in the high-CV setting, the fixed-batch window that survives the longest examples is often the throughput bottleneck.

\paragraph{The Distributed Group Alignment Problem.}
Moving batch formation to the point of accurate observability creates a distributed synchronization problem.
When each rank independently forms variable-size groups from its local realized lengths, the number of groups naturally differs across ranks, but DDP requires all ranks to call \texttt{AllReduce} the same number of times.
We call the task of aligning these runtime group counts without sample loss or deadlock the \textbf{Distributed Group Alignment Problem} (DGAP).
Prior length-aware batchers and static micro-batch/DDP systems~\cite{wolf2020transformers,ott2019fairseq,kuchaiev2019nemo,mosaicml2022composer,shoeybi2019megatron,paszke2019pytorch,li2020pytorch} avoid this regime by fixing batch composition or sharding statically; ODB's novelty is the runtime DDP-aware DataLoader regime, not another offline ordering heuristic.

\paragraph{Contributions.}
We make four contributions. \textbf{(i)} We introduce \textbf{Online Dynamic Batching (\odb)}, a DataLoader-side drop-in system that observes preprocessing-/augmentation-dependent post-pipeline lengths and requires no model, optimizer, attention-kernel, or length-cache precompute. \textbf{(ii)} We formalize DGAP and give a Max-Based Bidirectional Group Alignment protocol that provides DDP step alignment, strict identity coverage in default join mode, sample-quota closure in opt-in non-join mode, and deadlock-free bounded termination (Theorems~\ref{thm:strict-zero}--\ref{thm:deadlock-free}). \textbf{(iii)} We evaluate against Standard, Sorted, Packing, and offline GMT/BMT/HFG oracle baselines across 2B/8B Full FT and LoRA plus a production MM-Mix case study, showing 1.58--2.51$\times$ single-node and 1.71--3.78$\times$ two-node public throughput, plus 4.43$\times$ in the production case study, while keeping validation/benchmark metrics in the Standard-comparable band. \textbf{(iv)} We identify regimes where online observability matters---augmentation-policy churn, multimodal preprocessing, and high-CV long tails---and release the open-source \texttt{online-dynamic-batching} package with lightweight trainer adapters.

\section{System Design}
\label{sec:design}

\subsection{Architecture Overview}

ODB wraps the PyTorch \texttt{DataLoader} iterator boundary, leaving the Dataset and Model untouched; a lightweight trainer-side integration consumes emitted-sample metadata for accounting, token-level loss scaling, and optional sample-quota stopping (\S\ref{sec:implementation}; Figure~\ref{fig:architecture}).
Workers run with a \emph{null collate} (passing single samples), and a dedicated \emph{collate process} drains the worker queue into a configured grouping buffer (1024 samples in our experiments).
On each round it sorts by length, greedily forms variable-size batches (Section~\ref{sec:batch_sizing}), and synchronizes group counts via a single Gloo \texttt{all\_gather}, then runs split/overflow to align ranks to a common target $T_{\mathrm{grp}}$ (Section~\ref{sec:alignment}).
Aligned groups pass through the configured \texttt{collate\_fn} and are placed in the output queue; under-filled slots are padded with \texttt{IDLE\_DATA} sentinels that the main process transparently skips, leaving real batches step-aligned across active ranks.
The collate process owns its own Gloo group, fully isolated from the main process's NCCL group.

\begin{figure}[t]
    \centering
    \begin{subfigure}[t]{0.48\textwidth}
        \centering
        \scalebox{0.78}{
        \begin{tikzpicture}[
            node distance=0.3cm and 1.4cm,
            wbox/.style={rectangle, draw, rounded corners=2pt, fill=blue!8,
                minimum width=2.4cm, minimum height=0.5cm,
                font=\tiny, align=center, line width=0.5pt},
            cbox/.style={rectangle, draw, rounded corners=2pt, fill=blue!8,
                minimum width=2.4cm, minimum height=0.5cm,
                font=\tiny, align=center, line width=0.5pt},
            tbox/.style={rectangle, draw, rounded corners=2pt, fill=red!8,
                minimum width=2.4cm, minimum height=0.4cm,
                font=\tiny, align=center, line width=0.5pt},
            arr/.style={-{Stealth[length=2.5pt]}, line width=0.4pt},
            comm/.style={<->, dashed, line width=0.5pt},
            lbl/.style={font=\tiny, text=gray!60!black},
        ]
        \node[wbox] (w0) {Worker $\times n_w$\\[-1pt]
            {\tiny augment+template $\to$ tokenize}};
        \node[cbox, below=of w0] (c0) {collate\_fn($B$)\\[-1pt]
            {\tiny pad to $l_{\max}$}};
        \node[tbox, below=0.5cm of c0] (t0) {Trainer};
        \draw[arr] (w0) -- (c0);
        \draw[arr] (c0) -- node[right, lbl] {fixed bs=$B$} (t0);
        \node[font=\tiny\bfseries, above=0.1cm of w0] {Rank 0};

        \node[wbox, right=1.6cm of w0] (w1) {Worker $\times n_w$\\[-1pt]
            {\tiny augment+template $\to$ tokenize}};
        \node[cbox, below=of w1] (c1) {collate\_fn($B$)\\[-1pt]
            {\tiny pad to $l_{\max}$}};
        \node[tbox, below=0.5cm of c1] (t1) {Trainer};
        \draw[arr] (w1) -- (c1);
        \draw[arr] (c1) -- node[right, lbl] {fixed bs=$B$} (t1);
        \node[font=\tiny\bfseries, above=0.1cm of w1] {Rank 1};

        \draw[comm, red!60!black] (t0) -- node[above, font=\tiny\itshape,
            text=red!50!black] {NCCL} (t1);
        \end{tikzpicture}}
        \caption{Standard DataLoader}
        \label{fig:arch-standard}
    \end{subfigure}
    \hfill
    \begin{subfigure}[t]{0.48\textwidth}
        \centering
        \scalebox{0.78}{
        \begin{tikzpicture}[
            node distance=0.3cm and 1.4cm,
            wbox/.style={rectangle, draw, rounded corners=2pt, fill=blue!8,
                minimum width=2.4cm, minimum height=0.5cm,
                font=\tiny, align=center, line width=0.5pt},
            cbox/.style={rectangle, draw, rounded corners=2pt, fill=orange!12,
                minimum width=2.4cm, minimum height=0.5cm,
                font=\tiny, align=center, line width=0.5pt},
            tbox/.style={rectangle, draw, rounded corners=2pt, fill=red!8,
                minimum width=2.4cm, minimum height=0.4cm,
                font=\tiny, align=center, line width=0.5pt},
            arr/.style={-{Stealth[length=2.5pt]}, line width=0.4pt},
            comm/.style={<->, dashed, line width=0.5pt},
            lbl/.style={font=\tiny, text=gray!60!black},
        ]
        \node[wbox] (w0) {Worker $\times n_w$\\[-1pt]
            {\tiny augment+template $\to$ tokenize}};
        \node[cbox, below=of w0] (c0) {Collate Process\\[-1pt]
            {\tiny sort $\to$ group $\to$ align}};
        \node[tbox, below=0.5cm of c0] (t0) {Trainer};
        \draw[arr] (w0) -- node[right, lbl] {single samples} (c0);
        \draw[arr] (c0) -- node[right, lbl, text=orange!70!black,
            font=\tiny\bfseries] {bs=3} (t0);
        \node[font=\tiny\bfseries, above=0.1cm of w0] {Rank 0};

        \node[wbox, right=1.6cm of w0] (w1) {Worker $\times n_w$\\[-1pt]
            {\tiny augment+template $\to$ tokenize}};
        \node[cbox, below=of w1] (c1) {Collate Process\\[-1pt]
            {\tiny sort $\to$ group $\to$ align}};
        \node[tbox, below=0.5cm of c1] (t1) {Trainer};
        \draw[arr] (w1) -- node[right, lbl] {single samples} (c1);
        \draw[arr] (c1) -- node[right, lbl, text=orange!70!black,
            font=\tiny\bfseries] {bs=2} (t1);
        \node[font=\tiny\bfseries, above=0.1cm of w1] {Rank 1};

        \draw[comm, orange!70!black] (c0) -- node[above, font=\tiny\itshape,
            text=orange!50!black] {Gloo} (c1);

        \draw[comm, red!60!black] (t0) -- node[above, font=\tiny\itshape,
            text=red!50!black] {NCCL} (t1);
        \end{tikzpicture}}
        \caption{ODB DataLoader}
        \label{fig:arch-odb}
    \end{subfigure}
    \caption{Architecture. (a)~Standard collates inside workers with a fixed bs on all ranks. (b)~ODB collates in a dedicated process, groups by length, and aligns group counts via Gloo; each rank may have a different local bs while real batches remain step-aligned across active ranks.}
    \label{fig:architecture}
\end{figure}
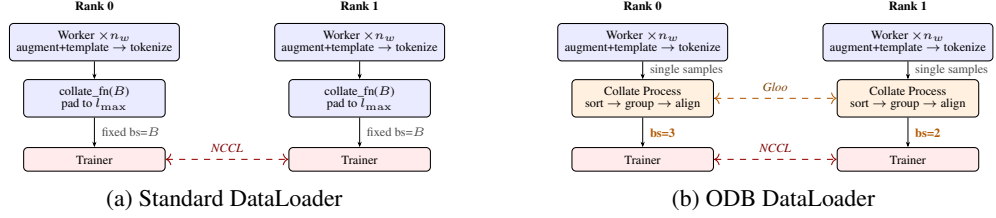

\subsection{Dynamic Batch Sizing}
\label{sec:batch_sizing}
ODB keeps per-batch token count roughly constant via a user-specified budget $L_{\max}$.
For a realized post-pipeline sample length $l$, $B(l)$ denotes the target local group size:
\begin{equation}
    B(l) = \max\left(\left\lfloor L_{\max}/l \right\rfloor, 1\right) \quad\text{so that}\quad B(l)\,l \approx L_{\max}.
    \label{eq:batch_size}
\end{equation}
Within each rank, buffered samples are sorted ascending and iterated from longest to shortest with a running group-size threshold $t$ (initially 1): each sample is appended to the current group, and when its size reaches $t$ the group is finalized and $t \leftarrow B(l)$ for the last-added (shortest) sample.
Successive groups naturally hold more samples since shorter $l$ yields larger $B(l)$, so per-group token counts converge to $L_{\max}$ (worked example in Appendix~\ref{app:grouping_example}).

\subsection{Cross-Rank Alignment, Termination, and Loss Scaling}
\label{sec:alignment}

\paragraph{Alignment target and bidirectional adjustment.}
Let $W$ be world size and $N$ the dataset identity count/full-epoch quota. Different ranks generally produce different group counts $G_r$. ODB computes a global group-count target $T_{\mathrm{grp}} = \max(\min(\max_{r:\,G_r>0} G_r, C_{\min}^{+}, S_{\min}^{+}), 1)$ over active ranks, where $C_{\min}^{+}$ is the minimum positive output-slot capacity on any active rank and $S_{\min}^{+}$ is the minimum positive buffered-sample count on any active rank. It then \emph{splits} groups by extracting singleton samples when $G_r<T_{\mathrm{grp}}$, or \emph{overflows} groups and recirculates extras when $G_r>T_{\mathrm{grp}}$, until every active rank reports $T_{\mathrm{grp}}$ groups. The full algorithm appears in Appendix~\ref{app:cross_rank_protocol}; \texttt{odb\_join\_mode} only changes termination: default join gives strict per-iteration identity coverage by draining outstanding sampler views before global completion, while opt-in non-join gives no-leak quota closure.

\begin{theorem}[Strict Zero-Discard, default join mode]
\label{thm:strict-zero}
Under \texttt{odb\_join\_mode = true} with \texttt{DistributedSampler(drop\_last=False)}, each rank emits its remaining and outstanding sampler views before advertising local finish, and the collate subprocess keeps participating in the shared Gloo protocol until all ranks advertise local finish. Therefore the emitted sampler-view multiset equals the sampler multiset $M=W\lceil N/W\rceil$, its identity projection covers all $N$ dataset identities, and per-iteration $\eta_{\text{logical}}=0$ by construction.
\end{theorem}

\begin{theorem}[No-Leak \& Sample-Quota Closure, opt-in non-join]
\label{thm:zero-discard}
Let $S_{\max}$ be the largest realized global emit count of one aligned trainer step, and let $S_{\text{emit}}$ be the global cumulative trainer-side emitted-sample count at termination. Under non-join termination, every sampler view remains in exactly one protocol state (emitted, collate buffer, worker queue, or sampler pending), and the trainer-side control logic chains logical iterations until $N \le S_{\text{emit}} \le N+S_{\max}$. Hence $\eta_{\text{quota}} := \max(0,1-S_{\text{emit}}/N)=0$. This is a cumulative sample-count guarantee; non-join identity audits are reported in App.~\ref{app:eta_identity}.
\end{theorem}

\begin{corollary}[Empirical sample-quota closure]
\label{cor:eta_zero}
$\eta_{\text{quota,emp}} = 0$ across all \textbf{18 evaluation runs} and 6 synthetic distributions (terminal epoch rounded to four decimals $\in \{1.0000,\, 1.0001\}$, with only final-step overshoot; App.~\ref{app:eta_measurement}).
\end{corollary}

\paragraph{Unified loop and termination.}
Let \texttt{pf} and \texttt{nw} denote the DataLoader prefetch factor and worker count, and let $D=\max(\texttt{pf}{\cdot}\texttt{nw},\texttt{buffer\_size})$ be the per-rank outstanding-depth envelope. A single \texttt{all\_gather} per round exchanges \texttt{[idx\_budget$_r$, n\_groups$_r$, sizes$_r$]} and, when token-level loss scaling is enabled, \texttt{tokens$_r$}, with \texttt{n\_groups$_r$} $\in\{n{>}0, 0, -1\}$ encoding ``produced $n$ groups'', ``insufficient data'', or ``finished''. In default join mode, ranks drain outstanding sampler views before the shared completion signal; in opt-in non-join termination, a logical \texttt{DistributedSampler} iteration ends when any rank emits $-1$ and the trainer-side control logic launches subsequent logical iterations until $S_{\text{emit}}\ge N$ (per early stop, at most $W{\cdot}D$ already-fetched sampler views are not delivered to the trainer; Lemma~\ref{lem:logical_bound}). Exact loss scaling is a separate accounting option: a deterministic all-rank predicate can trigger a second \texttt{all\_gather} to re-broadcast post-alignment token counts.

\begin{theorem}[Bounded Termination and Deadlock-Freedom]
\label{thm:deadlock-free}
Each logical iteration terminates in $O(N/W) + O(D)$ rounds and no rank blocks on \texttt{all\_gather}, given a finite sampler and uniform-call invariant (Appendix~\ref{app:proofs_uniform}).
\end{theorem}

Appendix~\ref{app:proofs} gives the per-rank transition rules and proofs for sample-quota closure and bounded termination; Appendix~\ref{app:state_machines} visualizes the non-join state machine; Appendix~\ref{app:join_tradeoff} proves the default join-mode identity contract and reports its throughput cost. The coordination channel uses a dedicated Gloo group inside the collate subprocess (isolated from NCCL), $\sim$128\,KB per round at $W{=}8$---orders of magnitude below gradient \texttt{AllReduce}---and overlaps GPU compute ($<2\%$ overhead). Appendix~\ref{app:multinode} additionally validates empty-rank deadlock-freedom and reports two-node 2B/8B Full FT results.

\paragraph{Loss scaling.}
ODB's per-rank batches differ in token counts $t_r$, so naive DDP averaging $\frac{1}{W}\sum_r \bar{\mathcal{L}}_r$ is a biased estimate of the per-token reference loss $\mathcal{L}^{\star}=\frac{1}{T_{\mathrm{tok}}}\sum_{r,i,k}\ell_{r,i,k}$ ($T_{\mathrm{tok}}{=}\sum_r t_r$).
Pre-scaling each rank's loss by $W \cdot w_r$ makes DDP's post-averaging output equal $\sum_r w_r \bar{\mathcal{L}}_r$; the unique weight that recovers $\mathcal{L}^{\star}$ bit-precisely is the \emph{token-level} weight
\begin{equation}
w_r = \frac{t_r}{T_{\mathrm{tok}}},\qquad \text{so}\quad \sum_r w_r \bar{\mathcal{L}}_r = \mathcal{L}^{\star}.
\label{eq:exactness}
\end{equation}
Sample-level weighting instead weights ranks by sample count; it matches the per-token reference $\mathcal{L}^{\star}$ only when the average tokens per sample $t_r/n_r$ is identical across ranks, a condition ODB's group-of-groups alignment generally does not impose. All main experiments therefore use token-level scaling; App.~\ref{app:loss_scaling} gives the full derivation and approximate-vs-exact comparison.

\subsection{API and Packaging}
\label{sec:implementation}
ODB is implemented as an in-place PyTorch \texttt{DataLoader} wrapper with no changes to the Model or Dataset. The reference LLaMA-Factory integration consumes ODB step metadata for emitted-sample accounting and token-level loss scaling; the same metadata interface can be ported to other Trainer frameworks. ODB is released as a standalone open-source Python package (\texttt{online-dynamic-batching}).

\section{Evaluation}
\label{sec:eval}

\subsection{Experimental Setup}

\paragraph{Hardware.}
Unless labeled otherwise, the main experiments use one 8$\times$H20 node (96\,GB/GPU), DeepSpeed ZeRO-2, and bf16. Labeled multi-node experiments use two 8$\times$H20 nodes (16 GPUs total) with the same ZeRO-2/bf16 training stack; these results and the empty-rank audit appear in Appendix~\ref{app:multinode}.

\paragraph{Datasets and models.}
We evaluate on three public datasets spanning text and multimodal modalities, reporting the CV defined in Section~\ref{sec:intro}:
\textbf{UltraChat-200K}~\cite{ding2023ultrachat} (text-only SFT, 208K, CV=0.48),
\textbf{LLaVA-150K}~\cite{liu2024visual} (multimodal, 158K, CV=0.29), and
\textbf{ShareGPT4o}~\cite{li2024llavaonevision} (multimodal, 57K, CV=1.00; the GPT-4o-curated subset distributed with LLaVA-OneVision), plus six synthetic distributions for correctness audits.
We use \textbf{Qwen3-VL-2B/8B-Instruct} checkpoints~\cite{bai2025qwen3vl} under both Full FT and LoRA (rank=8, target=all); earlier Qwen-VL reports document the lineage of the released stack~\cite{bai2025qwen25vl,bai2023qwen}.
Per-dataset \texttt{cutoff\_len} is above the observed maximum (UltraChat 8192/Max 4471; LLaVA 2048/Max 1260; ShareGPT4o 16384/Max 12110), ensuring \emph{zero truncation} and shared across methods (Table~\ref{tab:length_stats}).
Standard, Sorted, and ODB are mathematically insensitive to \texttt{cutoff\_len} above the longest realized sample, whereas Packing's per-step memory cost and oracle max-token feasibility depend on it; keeping a uniform \texttt{cutoff\_len} is therefore essential for a fair Packing/oracle comparison.

\paragraph{Baselines and training hyperparameters.}
Baselines are \textbf{Standard} (fixed batch size, random sampling), \textbf{Sorted} (online length-grouped fixed batch), \textbf{Packing} (HuggingFace sequence packing; text-only in our stack, \S\ref{sec:related})~\cite{krell2021efficient}, \textbf{GMT-/BMT-oracle} fairseq-style global/bucketed max-token batchers~\cite{ott2019fairseq} with a one-time scalar cache of post-pipeline \texttt{len(input\_ids)}, \textbf{HFG-oracle} randomized fixed-batch \texttt{group\_by\_length}, and \textbf{ODB} (ours).
Oracle caches are used only for batch construction: training still executes the same online preprocessing, augmentation, templating, tokenization, and visual-token expansion path. The cache is per-(dataset, transform policy, template, \texttt{cutoff\_len}), so it must be rebuilt when those policies change; precompute cost is reported in Appendix~\ref{app:setup_details}.
All methods use AdamW+cosine, lr $10^{-5}$, \texttt{warmup\_ratio}=0.03, grad-clip 4.0, one epoch, bf16+ZeRO-2, and no gradient accumulation.

\paragraph{Method-specific parameters and configuration search.}
For Standard/Sorted, the throughput knob is \texttt{batch\_size} (bs); we sweep $\{1,2,4,8,16\}$ where memory permits, and additionally profile larger values to establish OOM/full-epoch survival when noted.
For ODB, per-worker bs is always 1; the throughput knobs are $L_{\max}$ (per-step token budget, $B(l)=\max(\lfloor L_{\max}/l\rfloor,1)$) and \texttt{prefetch\_factor} (\texttt{pf}, with $D=\max(\texttt{pf}{\times}\texttt{nw},\texttt{buffer\_size})$); we sweep $L_{\max}$ over a feasibility-tested token-budget grid up to 32768, including powers-of-two anchors and intermediate tuned values, then \texttt{pf} at the selected budget.
For GMT-/BMT-oracle, we sweep \texttt{max\_tokens\_budget} over the corresponding feasible token-budget grid; HFG-oracle sweeps fixed bs under the same survival protocol used for Sorted.
Each search runs for 20 minutes to obtain steady-state throughput, and a candidate is eligible only if it produces stable numeric throughput without OOM in the profiling window.
HP selection uses training-only profiling signals: completed non-OOM candidates are restricted to a near-fastest throughput band, then deterministic stability/resource tie-breakers are applied; validation loss and benchmark scores are never used.
The provisional selection must also complete the final full-epoch training run; if it OOMs or otherwise fails to finish the full epoch, we fall back to the next-ranked eligible configuration under the same profiling rule.
\textbf{Fairness principle.} Quality is measured at each method's selected full-epoch-surviving configuration, not at an arbitrary shared setting; for Sorted, the selected bs must also complete a full epoch (Table~\ref{tab:public_fullft}, footnote~\textsuperscript{c}).

\paragraph{Quality benchmarks.}
UltraChat uses \textbf{MMLU}~\cite{hendrycks2021mmlu}; multimodal tasks use \textbf{MMMU-MC}~\cite{yue2024mmmu}, a parser-free choice-likelihood score over letter-labeled multiple-choice validation items.
MMMU-MC excludes non-letter ground-truth rows and avoids generation/parsing by scoring the assistant first-token likelihood of option letters A--H; generated-answer or parser-based analyses are never mixed into the Score column.

\subsection{Main Results}
\label{sec:e2e}

\begin{table}[!t]
\centering
\caption{\textbf{Full Fine-Tuning Results.} 8B/2B Full FT on 8$\times$H20, selected config per method, 3-seed mean$\pm$std. \texttt{sam/s} is emitted samples divided by wall-clock. Packing is text-only here; GMT/BMT/HFG use scalar length caches for batch construction and exclude cache construction (App.~\ref{app:setup_details}). Bold marks the highest emitted-sample throughput among online/no-cache rows; offline oracles are unbolded comparators. Score is MMLU for UltraChat and MMMU-MC choice likelihood for multimodal tasks.}
\label{tab:public_fullft}
\tiny
\setlength{\tabcolsep}{1pt}
\resizebox{\linewidth}{!}{%
\begin{tabular}{@{}llcccc@{\hspace{4pt}}cccc@{}}
\toprule
 & & \multicolumn{4}{c}{\textbf{8B (Full FT)}} & \multicolumn{4}{c}{\textbf{2B (Full FT)}} \\
\cmidrule(lr){3-6} \cmidrule(lr){7-10}
Dataset & Method & sam/s & Speedup & Score & Val Loss & sam/s & Speedup & Score & Val Loss \\
\midrule

\multirow{7}{*}{\shortstack[l]{UltraChat\\CV=0.48\\MMLU}}
& Standard     & $5.77 \pm 0.01$ & 1.00$\times$ & $72.85 \pm 0.40$\% & $0.8487 \pm 0.0013$ & $20.98 \pm 0.04$ & 1.00$\times$ & $58.17 \pm 0.06$\% & $0.9980 \pm 0.0017$ \\
& Sorted       & $8.09 \pm 0.02$\textsuperscript{c} & 1.40$\times$ & $74.72 \pm 0.33$\% & $0.8839 \pm 0.0020$ & $28.31 \pm 0.08$ & 1.35$\times$ & $59.02 \pm 0.10$\% & $1.0287 \pm 0.0012$ \\
& Packing      & $10.46 \pm 0.00$ & 1.81$\times$ & $75.18 \pm 0.06$\% & $1.1819 \pm 0.0088$\textsuperscript{d} & $36.61 \pm 0.07$ & 1.75$\times$ & $59.68 \pm 0.07$\% & $1.1947 \pm 0.0008$\textsuperscript{d} \\
& GMT-oracle\textsuperscript{f}   & $10.94 \pm 0.02$ & 1.90$\times$ & $75.14 \pm 0.09$\% & $0.8350 \pm 0.0013$ & $39.44 \pm 0.03$ & 1.88$\times$ & $59.16 \pm 0.09$\% & $0.9987 \pm 0.0020$ \\
& BMT-oracle\textsuperscript{f}   & $10.31 \pm 0.02$ & 1.79$\times$ & $75.26 \pm 0.08$\% & $0.8351 \pm 0.0013$ & $35.84 \pm 0.04$ & 1.71$\times$ & $59.15 \pm 0.16$\% & $0.9772 \pm 0.0013$ \\
& HFG-oracle\textsuperscript{f}   & $7.33 \pm 0.01$ & 1.27$\times$ & $73.70 \pm 0.14$\% & $0.8404 \pm 0.0013$ & $27.28 \pm 0.07$ & 1.30$\times$ & $58.30 \pm 0.13$\% & $0.9974 \pm 0.0016$ \\
& \textbf{ODB} & $\mathbf{10.23 \pm 0.03}$\textsuperscript{b} & \textbf{1.77$\times$} & $74.75 \pm 0.11$\% & $0.8558 \pm 0.0014$ & $\mathbf{36.91 \pm 0.19}$ & \textbf{1.76$\times$} & $58.98 \pm 0.18$\% & $1.0030 \pm 0.0014$ \\
\midrule
\multirow{6}{*}{\shortstack[l]{LLaVA\\CV=0.29\\MMMU-MC}}
& Standard     & $14.38 \pm 0.03$ & 1.00$\times$ & $55.88 \pm 0.72$\% & $1.0552 \pm 0.0013$ & $47.92 \pm 0.05$ & 1.00$\times$ & $43.06 \pm 0.77$\% & $1.1814 \pm 0.0225$ \\
& Sorted       & $20.46 \pm 0.03$\textsuperscript{c} & 1.42$\times$ & $55.84 \pm 0.38$\% & $1.1781 \pm 0.0028$ & $66.37 \pm 0.07$ & 1.39$\times$ & $39.53 \pm 0.54$\% & $1.3318 \pm 0.0017$ \\
& GMT-oracle\textsuperscript{f}   & $26.65 \pm 0.05$ & 1.85$\times$ & $53.53 \pm 0.66$\% & $1.0630 \pm 0.0011$ & $79.44 \pm 0.39$ & 1.66$\times$ & $43.06 \pm 0.40$\% & $1.2102 \pm 0.0018$ \\
& BMT-oracle\textsuperscript{f}   & $25.70 \pm 0.05$ & 1.79$\times$ & $54.24 \pm 0.65$\% & $1.0630 \pm 0.0012$ & $75.64 \pm 0.13$ & 1.58$\times$ & $43.49 \pm 0.18$\% & $1.2102 \pm 0.0013$ \\
& HFG-oracle\textsuperscript{f}   & $21.84 \pm 0.04$ & 1.52$\times$ & $54.71 \pm 0.65$\% & $1.0590 \pm 0.0012$ & $69.52 \pm 0.08$ & 1.45$\times$ & $42.59 \pm 0.96$\% & $1.2007 \pm 0.0205$ \\
& \textbf{ODB} & $\mathbf{24.87 \pm 0.09}$\textsuperscript{b} & \textbf{1.73$\times$} & $54.08 \pm 0.53$\% & $1.0944 \pm 0.0013$ & $\mathbf{82.42 \pm 0.53}$ & \textbf{1.72$\times$} & $43.18 \pm 0.31$\% & $1.2189 \pm 0.0021$ \\
\midrule
\multirow{6}{*}{\shortstack[l]{ShareGPT4o\\CV=1.00\\MMMU-MC}}
& Standard     & $2.37 \pm 0.00$\textsuperscript{a} & 1.00$\times$ & $52.43 \pm 0.44$\% & $1.1913 \pm 0.0056$ & $6.51 \pm 0.01$ & 1.00$\times$ & $41.29 \pm 0.71$\% & $1.2910 \pm 0.0058$ \\
& Sorted       & $2.44 \pm 0.00$\textsuperscript{a} & 1.03$\times$ & $52.82 \pm 0.31$\% & $1.2175 \pm 0.0059$ & $6.71 \pm 0.01$ & 1.03$\times$ & $39.33 \pm 0.37$\% & $1.3191 \pm 0.0062$ \\
& GMT-oracle\textsuperscript{f}   & $2.57 \pm 0.01$ & 1.09$\times$ & $53.57 \pm 0.71$\% & $1.1904 \pm 0.0057$ & $7.03 \pm 0.02$ & 1.08$\times$ & $41.18 \pm 0.12$\% & $1.2930 \pm 0.0059$ \\
& BMT-oracle\textsuperscript{f}   & $2.50 \pm 0.01$ & 1.06$\times$ & $52.51 \pm 0.88$\% & $1.1908 \pm 0.0055$ & $7.03 \pm 0.00$ & 1.08$\times$ & $40.79 \pm 0.24$\% & $1.2931 \pm 0.0059$ \\
& HFG-oracle\textsuperscript{f}   & $2.82 \pm 0.01$ & 1.19$\times$ & $52.31 \pm 0.76$\% & $1.1913 \pm 0.0058$ & $7.71 \pm 0.04$ & 1.18$\times$ & $40.83 \pm 1.31$\% & $1.2909 \pm 0.0061$ \\
& \textbf{ODB} & $\mathbf{5.83 \pm 0.04}$\textsuperscript{b} & \textbf{2.46$\times$} & $53.88 \pm 0.20$\% & $1.2269 \pm 0.0059$ & $\mathbf{16.09 \pm 0.21}$ & \textbf{2.47$\times$} & $40.12 \pm 0.36$\% & $1.3341 \pm 0.0062$ \\
\bottomrule
\end{tabular}}

\vspace{1mm}
\begin{minipage}{\linewidth}
\raggedright\scriptsize
\textsuperscript{a}For ShareGPT4o 8B, Standard and Sorted use bs=1; wider fixed batches are infeasible or slower on the long-tail length distribution.
\textsuperscript{b}ODB's $(L_{\max}, \texttt{pf}, \texttt{buffer})$ is selected per dataset; tuples appear in App.~\ref{app:setup_details}.
\textsuperscript{c}Sorted bs is the largest value that completes a full epoch; larger profiled settings OOM on longest-tail batches.
\textsuperscript{d}Packing val\_loss uses a packed-sequence denominator and is not comparable to per-sample val\_loss; MMLU is comparable.
\textsuperscript{e}MMMU-MC reports the 850/900 rows with A--H letter ground truth; generated-answer analyses are not mixed into Score.
\textsuperscript{f}Oracle baselines use scalar length caches for batch construction; reported throughput excludes cache construction.
\end{minipage}
\end{table}

\paragraph{Deployment-class reading.}
The main comparison is whether an online DataLoader-level method can approach stronger offline/model-side throughput while preserving multimodal deployment properties. Packing is text-only in our stack; GMT/BMT/HFG are favorable oracle comparators with exact post-pipeline scalar lengths, no cache cost charged in the throughput column, and the same runtime preprocessing path during training. ODB instead observes post-pipeline lengths online, requires no cache, and achieves the highest emitted-sample throughput among online/no-cache rows on every 8B Full FT dataset.

\paragraph{8B Full FT.}
ShareGPT4o isolates the long-tail case: ODB reaches $2.46\times$, while Sorted remains at $1.03\times$ because the longest examples force bs=1. Offline oracles remain sample-count limited by construction rather than by missing length visibility: HFG keeps a fixed batch size, and GMT/BMT constrain padded token area, yielding only 11.4/9.0 samples per update versus ODB's 52.8 (Table~\ref{tab:throughput_decomp_8b}). On LLaVA, ODB reaches $1.73\times$ with low padding and higher throughput than Sorted, while avoiding the Sorted validation-loss and answer-format sensitivity discussed in \S\ref{sec:quality}; on UltraChat, Packing and GMT-oracle are strong non-drop-in/offline comparators, while ODB is the highest-throughput online/no-cache row.

\paragraph{2B Full FT and LoRA.}
At 2B, ODB again leads the online/no-cache rows: it reaches $2.47\times$ on ShareGPT4o, $1.72\times$ on LLaVA, and $1.76\times$ on UltraChat. The narrow exceptions relative to all comparators clarify the boundary of the claim: GMT-oracle is faster on 2B UltraChat because global offline construction can use full-dataset length information, whereas ODB greedily groups within an online buffer; on 2B LLaVA, ODB is both the fastest and the lowest-padding online/drop-in row. Under LoRA, ODB reaches $1.58$--$2.51\times$ (App.~\ref{app:lora_results}). Appendix~\ref{app:multinode} adds two-node Full FT validation, where ODB remains the highest-throughput online/no-cache row and reaches $1.71$--$3.78\times$.

\paragraph{Takeaway.}
ODB brings online/drop-in batching into the oracle token-budget regime. The remaining raw-throughput wins require assumptions outside this deployment class: text-only packing support, offline length-cache precompute, full-dataset visibility, or cache rebuilds under augmentation/template/cutoff changes.

\subsection{Training Quality}
\label{sec:quality}

\noindent\textbf{Quality is evaluated at the speed-selected configuration.}
ODB intentionally changes batch shape: its efficiency mechanism is to increase useful tokens and samples per optimizer update while controlling padding.
For every method, the reported score is therefore measured at the configuration selected by the same training-only profiling protocol, which restricts completed non-OOM candidates to a near-fastest throughput band and then applies deterministic stability/resource tie-breakers.
This makes the quality comparison a direct test of whether the faster update geometry harms the trained model.

\paragraph{Discriminative case: 8B Full FT LLaVA.}
ODB stays in the Standard-comparable operating band, though not tied with Standard on every metric: val\_loss $1.0944\pm0.0013$ vs. Standard $1.0552\pm0.0013$ and MMMU-MC $54.08\pm0.53\%$ vs. $55.88\pm0.72\%$, while delivering $1.73\times$ throughput. Sorted reaches $1.42\times$ but regresses by $+11.6\%$ on val\_loss. Appendix~\ref{app:sorted_format} explains why generated-answer MMMU is not mixed into the main benchmark: length-sorted training can bias answer format, whereas ODB groups only within each online buffer.

\paragraph{Other quality cells.}
On 8B UltraChat, ODB MMLU $74.75\pm0.11\%$ remains in the same band as Packing ($75.18\pm0.06\%$), GMT ($75.14\pm0.09\%$), and BMT ($75.26\pm0.08\%$); val\_loss remains within 0.008 of Standard and below Sorted. On 2B Full FT, ODB reports UltraChat MMLU $+0.81$ pp over Standard, LLaVA MMMU-MC $+0.12$ pp, and ShareGPT4o MMMU-MC $-1.17$ pp. Under LoRA, ODB remains within the Standard band on LLaVA at both scales and is nominally above Standard on ShareGPT4o at both scales (App.~\ref{app:lora_results}).

Oracle rows are offline, cache-based comparators, not direct online baselines. Their raw-throughput differences reflect update geometry rather than a different data path: on 8B LLaVA, GMT/BMT use about 251--254 samples/update and 591--597 updates/epoch, while ODB uses 177.6 samples/update and 844 updates; on 2B LLaVA, ODB is both faster and uses 119.0 samples/update versus about 60 for GMT/BMT (Tables~\ref{tab:throughput_decomp_8b}--\ref{tab:throughput_decomp_2b}). We therefore read all rows as throughput--quality operating points, with ODB adding online/drop-in batching, DDP step alignment, sample-quota closure, and default join-mode identity coverage.

\subsection{Speedup Analysis}
\label{sec:cv_analysis}

\noindent\textbf{ODB speedup comes from jointly improving the three efficiency
conditions in Section~\ref{sec:intro}:} spatial efficiency (little padding),
compute saturation (enough useful token/FLOPs per step), and temporal efficiency
(overlapping input preparation with GPU compute).
Fixed-batch training on
variable-length data cannot optimize all three simultaneously: bs=1 avoids
padding but underfills the GPU, while larger bs increases per-step work only
by mixing unequal lengths and therefore pays padding or OOM cost. ODB improves
throughput because it attacks the three terms together: online length grouping
reduces spatial waste, token-budget updates increase useful work per step, and
bounded outstanding depth hides the remaining input latency.

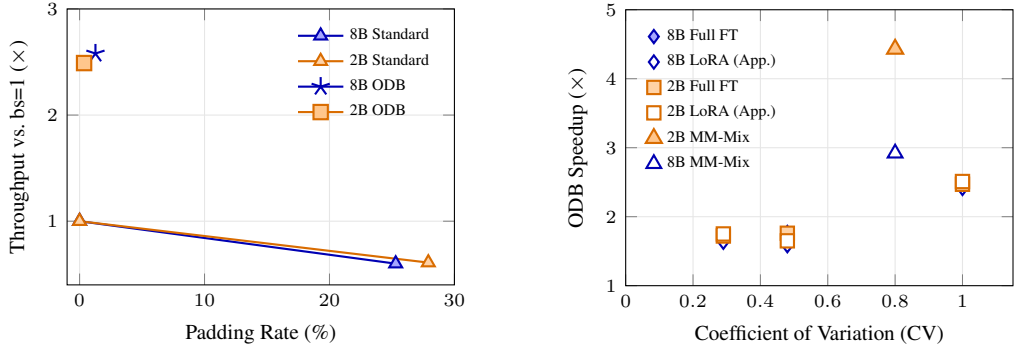
\begin{figure}[!t]
\centering
\begin{subfigure}[t]{0.48\textwidth}
    \centering
    \begin{tikzpicture}
    \begin{axis}[
        width=\linewidth,
        height=0.78\linewidth,
        xlabel={Padding Rate (\%)},
        ylabel={Throughput vs.\ bs=1 ($\times$)},
        xmin=-1, xmax=30,
        ymin=0.4, ymax=3.0,
        grid=major,
        grid style={gray!20},
        legend pos=north east,
        legend style={font=\tiny, draw=none, fill=none, cells={anchor=west}},
        mark size=2.8pt,
        every axis label/.style={font=\footnotesize},
        tick label style={font=\scriptsize},
    ]
    \addplot[mark=triangle*, blue!70!black, thick, mark options={fill=blue!35}] coordinates {
        (0, 1.00)
        (25.3, 0.60)
    };
    \addlegendentry{8B Standard}
    \addplot[mark=triangle*, orange!85!black, thick, mark options={fill=orange!35}] coordinates {
        (0, 1.00)
        (27.91, 0.61)
    };
    \addlegendentry{2B Standard}
    \addplot[mark=star, blue!70!black, mark size=3.6pt, thick, mark options={fill=blue!25}] coordinates {
        (1.27, 2.58)
    };
    \addlegendentry{8B ODB}
    \addplot[mark=square*, orange!85!black, mark size=2.7pt, thick, mark options={fill=orange!45}] coordinates {
        (0.36, 2.49)
    };
    \addlegendentry{2B ODB}
    \end{axis}
    \end{tikzpicture}
    \caption{ShareGPT4o 20-min profiling comparison (CV=1.00), normalized to fixed-bs=1; Table~\ref{tab:public_fullft} reports the selected 3-seed full-epoch rows. Fixed-bs=2 reaches 25--28\% padding and $\approx$0.60$\times$ throughput, while the selected ODB HP points keep padding at 1.27\% (8B) and 0.36\% (2B) with 2.58$\times$/2.49$\times$ profiling throughput.}
    \label{fig:padding_frontier}
\end{subfigure}
\hfill
\begin{subfigure}[t]{0.48\textwidth}
    \centering
    \begin{tikzpicture}
    \begin{axis}[
        width=\linewidth,
        height=0.78\linewidth,
        xlabel={Coefficient of Variation (CV)},
        ylabel={ODB Speedup ($\times$)},
        xmin=0, xmax=1.15,
        ymin=1.0, ymax=5.0,
        grid=major,
        grid style={gray!20},
        legend pos=north west,
        legend style={font=\tiny, draw=none, fill=none, cells={anchor=west}},
        mark size=2.5pt,
        every axis label/.style={font=\footnotesize},
        tick label style={font=\scriptsize},
    ]
    \addplot[only marks, mark=diamond*, blue!70!black, thick, mark options={fill=blue!35}] coordinates {
        (0.29, 1.73)
        (0.48, 1.77)
        (1.00, 2.46)
    };
    \addlegendentry{8B Full FT}
    \addplot[only marks, mark=diamond*, blue!70!black, thick, mark options={fill=white}] coordinates {
        (0.29, 1.63)
        (0.48, 1.58)
        (1.00, 2.41)
    };
    \addlegendentry{8B LoRA (App.)}
    \addplot[only marks, mark=square*, orange!85!black, thick, mark options={fill=orange!35}] coordinates {
        (0.29, 1.72)
        (0.48, 1.76)
        (1.00, 2.47)
    };
    \addlegendentry{2B Full FT}
    \addplot[only marks, mark=square*, orange!85!black, thick, mark options={fill=white}] coordinates {
        (0.29, 1.75)
        (0.48, 1.65)
        (1.00, 2.51)
    };
    \addlegendentry{2B LoRA (App.)}
    \addplot[only marks, mark=triangle*, orange!85!black, thick, mark size=3.5pt, mark options={fill=orange!45}] coordinates {
        (0.80, 4.43)
    };
    \addlegendentry{2B MM-Mix}
    \addplot[only marks, mark=triangle*, blue!70!black, thick, mark size=3.2pt, mark options={fill=white}] coordinates {
        (0.80, 2.92)
    };
    \addlegendentry{8B MM-Mix}
    \end{axis}
    \end{tikzpicture}
    \caption{ODB speedup vs. CV for main-table ODB rows plus MM-Mix markers. CV alone is insufficient: MM-Mix (CV=0.80, $f_s\approx0.37$) exceeds ShareGPT4o (CV=1.00, $f_s\approx0.01$).}
    \label{fig:cv_speedup}
\end{subfigure}
\caption{ODB speedup mechanisms. (a) On high-CV ShareGPT4o, fixed
batching moves rightward/downward as bs grows, while ODB occupies
selected low-padding/high-throughput points. (b) Across workloads, speedup tracks length
heterogeneity but is amplified by short-sample mass; LoRA points are appendix
context.}
\label{fig:speedup_views}
\end{figure}

\noindent\textbf{Spatial efficiency.}
Figure~\ref{fig:padding_frontier} shows the spatial effect on ShareGPT4o: fixed bs=2 introduces 25--28\% padding and drops to $\approx$0.60$\times$ of bs=1, while the selected ODB HP markers stay near the low-padding/high-throughput corner. This profiling view is consistent with the selected full-epoch rows in Table~\ref{tab:public_fullft}, where ODB reaches 2.46$\times$/2.47$\times$ on 8B/2B. The gain is therefore not merely larger batches, but larger \emph{useful} batches.

\noindent\textbf{Compute saturation.}
Fixed bs=1 avoids padding on long-tail workloads but gives the GPU too little useful work per update; larger fixed batches add work by adding padding. ODB instead increases real samples/tokens per update under a token budget, so the GPU sees denser useful work without the fixed-batch padding penalty. The effect is nearly scale-invariant within each dataset: UltraChat reaches 1.77$\times$/1.76$\times$ on 8B/2B, LLaVA reaches 1.73$\times$/1.72$\times$, and ShareGPT4o reaches 2.46$\times$/2.47$\times$. This pattern suggests that model scale is not the primary driver in these cells; Tables~\ref{tab:throughput_decomp_8b}--\ref{tab:throughput_decomp_2b} relate the throughput differences to batch-shape statistics.

\noindent\textbf{Short-sample leverage.}
CV flags padding pressure, but short-sample mass $f_s$ flags recoverable compute density. MM-Mix has lower CV than ShareGPT4o (0.80 vs. 1.00) but much larger $f_s$ ($\approx0.37$ vs. $\approx0.01$), allowing ODB to aggregate short OCR/VQA and captioning examples into compute-dense updates and reach 4.43$\times$ at 2B and 2.92$\times$ at 8B. $L_{\max}$ raises useful work per step until memory/step time saturate, while $D$ hides input latency until pipeline overlap saturates; Section~\ref{sec:ablation} locates these operating ranges.

\subsection{Ablation Study}
\label{sec:ablation}

Table~\ref{tab:ablation_ml} varies $L_{\max}$ at fixed $D{=}1024$ to expose batch-shape saturation; Table~\ref{tab:ablation_D} varies $D=\max(\texttt{pf} \times \texttt{nw}, \texttt{buffer\_size})$ to expose temporal overlap. Both are single-seed 20-minute profiling windows with default join-mode ODB; full-epoch throughput and quality claims remain in Table~\ref{tab:public_fullft}.

\begin{table}[!t]
\centering
\caption{Ablation: per-batch token budget $L_{\max}$ at fixed $D{=}1024$ (default join-mode ODB bs=1, \texttt{nw}=4, \texttt{pf}=256, \texttt{buffer}=1024, 8$\times$H20, Qwen3-VL-8B Full FT, single-seed 20-minute windows). Speedups use the Table~\ref{tab:public_fullft} 8B Standard baselines. Bold marks the fastest stable row; failed denotes no stable numeric throughput in the profiling window.}
\label{tab:ablation_ml}
\footnotesize
\setlength{\tabcolsep}{4pt}
\begin{tabular}{@{}lcccccc@{}}
\toprule
& \multicolumn{2}{c}{UltraChat (CV=0.48)} & \multicolumn{2}{c}{LLaVA (CV=0.29)} & \multicolumn{2}{c}{ShareGPT4o (CV=1.00)} \\
\cmidrule(lr){2-3} \cmidrule(lr){4-5} \cmidrule(lr){6-7}
$L_{\max}$ & sam/s & spd & sam/s & spd & sam/s & spd \\
\midrule
2048  & 8.48 & 1.47$\times$ & 20.25 & 1.41$\times$ & 5.48 & 2.31$\times$ \\
4096  & 9.23 & 1.60$\times$ & 22.66 & 1.58$\times$ & 5.74 & 2.42$\times$ \\
8192  & 9.44 & 1.64$\times$ & 24.17 & 1.68$\times$ & 5.96 & 2.52$\times$ \\
12288 & \textbf{10.21} & \textbf{1.77$\times$} & 24.53 & 1.71$\times$ & 6.11 & 2.58$\times$ \\
14336 & 10.08 & 1.75$\times$ & 24.44 & 1.70$\times$ & \textbf{6.17} & \textbf{2.60$\times$} \\
16384 & 9.90 & 1.72$\times$ & \textbf{24.88} & \textbf{1.73$\times$} & 5.86 & 2.47$\times$ \\
32768 & failed & --- & failed & --- & failed & --- \\
\bottomrule
\end{tabular}
\end{table}

\noindent\textbf{Reading Table~\ref{tab:ablation_ml}.}
Throughput rises as $L_{\max}$ fills each update, then regresses once memory pressure and longer steps dominate. The best stable $L_{\max}$ in the 8B sweep differs by dataset: $12288$ for UltraChat, $16384$ for LLaVA, and $14336$ for ShareGPT4o. The $32768$ setting is outside the stable profiling envelope. Appendix~\ref{app:quality_hp} gives the low-CV quality sensitivity; main-table configurations additionally tune $D$ and validate full-epoch quality.

\paragraph{Outstanding depth.}
After $L_{\max}$ fixes batch shape, $D$ controls how much prepared work can overlap GPU compute; at \texttt{nw}=4, \texttt{pf}<256 is clamped to $D{=}1024$ by buffer filling (App.~\ref{app:clamp}).

\begin{table}[!t]
\centering
\caption{Ablation: outstanding depth $D$ (default join-mode ODB, 8$\times$H20, Qwen3-VL Full FT, single-seed 20-minute windows; \texttt{nw}=4, \texttt{buffer}=1024). For each scale/dataset, $L_{\max}$ is fixed to the selected ODB budget used for the corresponding main-table cell. \texttt{ovrlap} denotes \texttt{pipeline\_overlap} in percent; bold uses unrounded \texttt{sam/s}.}
\label{tab:ablation_D}
\footnotesize
\setlength{\tabcolsep}{2pt}
\begin{tabular}{@{}lcccccccc@{}}
\toprule
& \multicolumn{2}{c}{$D{=}1024$} & \multicolumn{2}{c}{$D{=}2048$} & \multicolumn{2}{c}{$D{=}4096$} & \multicolumn{2}{c}{$D{=}8192$} \\
\cmidrule(lr){2-3} \cmidrule(lr){4-5} \cmidrule(lr){6-7} \cmidrule(lr){8-9}
Scale/Dataset & sam/s & ovrlap & sam/s & ovrlap & sam/s & ovrlap & sam/s & ovrlap \\
\midrule
2B UltraChat & 36.12 & 99.6 & 36.21 & 100.0 & \textbf{36.24} & 100.0 & 36.18 & 100.0 \\
2B LLaVA & 78.65 & 94.7 & \textbf{82.64} & 100.0 & 81.28 & 100.0 & 80.13 & 100.0 \\
2B ShareGPT4o & 15.77 & 97.0 & \textbf{15.80} & 100.0 & 15.46 & 100.0 & 14.17 & 100.0 \\
8B UltraChat & 9.93 & 99.9 & 9.93 & 100.0 & \textbf{9.93} & 100.0 & 9.92 & 100.0 \\
8B LLaVA & 24.65 & 98.5 & \textbf{24.85} & 100.0 & 24.85 & 100.0 & 24.68 & 100.0 \\
8B ShareGPT4o & \textbf{6.06} & 100.0 & 5.90 & 100.0 & 5.86 & 100.0 & 5.73 & 100.0 \\
\bottomrule
\end{tabular}
\end{table}

\noindent\textbf{Reading Table~\ref{tab:ablation_D}.}
At $D{=}1024$, overlap is already 94.7--100.0\%; $D{=}2048$ helps mainly on LLaVA, and larger values are flat or regressive once overlap saturates. Once overlap saturates or throughput stops improving, increasing $D$ adds little benefit and can be left at the smaller setting.

\paragraph{Buffer and loss defaults.}
On ShareGPT4o, most buffer gains appear by \texttt{buffer}=500; 1024--2000 is the high-throughput range for 2B (16.77--17.10\,sam/s) and 1024 for 8B (6.38\,sam/s), with padding at or below 0.6\% at \texttt{buffer}$\geq$1024 (App.~\ref{app:more_ablations}). The three loss-scaling modes have similar throughput in short profiling windows, within about 0.2\% on 2B and 1.0\% on 8B relative to sample-level scaling; tune $L_{\max}$ and $D$, not loss scaling.

\subsection{DataLoader I/O Sensitivity}
\label{sec:io_sensitivity}

Because ODB raises per-step sample count, input demand differs from fixed-bs Standard. The full \texttt{num\_workers} sweep (App.~\ref{app:io_sensitivity}, Table~\ref{tab:io}) shows the reported cells are not primarily worker-starved: ODB remains above Standard at every worker count, and by \texttt{nw}=4 rows are within about 5\% of their best values. We use \texttt{nw}=4 as a portable default and tune \texttt{pf} from 256.

\subsection{Case Study: Production Deployment}
\label{sec:case_study}

On production MM-Mix (273K samples from 7 OCR/VQA/captioning corpora, $\approx$545K sample-views over 2 epochs; CV${\approx}0.8$, bimodal; App.~\ref{app:setup_details}), two-node Qwen3-VL-2B full-epoch runs reach Standard at $17.85\pm0.15$\,sam/s, Sorted at $20.62\pm2.08$\,sam/s ($1.15\times$), and \textbf{ODB $L_{\max}{=}12288$ at $79.15\pm4.16$\,sam/s ($4.43\times$)}---our largest speedup (Table~\ref{tab:mmmix_quality}). The high short-sample fraction ($f_s{\approx}0.37$ vs.\ ShareGPT4o's 0.01) activates the compute-density mechanism: ODB aggregates OCR/VQA short examples into dense batches, exceeding what CV alone would suggest; the 8B profiling sweep shows the same qualitative pattern (\textbf{22.07\,sam/s, $2.92\times$}).

The corresponding full-epoch rows show Standard-comparable parser-free MMMU-MC behavior: ODB reaches $46.31\pm0.44$, while Standard is $43.33\pm2.24$ and Sorted is $43.65\pm2.32$. Its validation loss is higher than Standard but far below Sorted's degradation, and the benchmark score remains in the Standard-comparable range. GMT/BMT-oracle obtain higher MMMU-MC scores than ODB, but their validation loss is not lower than Standard's; we treat this as benchmark-specific variation and retain them as offline cache-based comparators. The measured MM-Mix oracle-cache construction alone takes 299.9\,s for 272{,}589 samples on one H20, a favorable churn-inclusive lower bound that excludes sample-store scans, order/materialization, cache validation, distributed staging, and rebuilds after recipe changes (App.~\ref{app:setup_details}).

\section{Practical Guidance}
\label{sec:guidance}

\textbf{ROI.} Estimate CV and $f_s=\Pr[l<L_{\max}/4]$ on a 1k--5k sketch: CV$\gtrsim0.8$ or $f_s>0.2$ flags high-ROI regimes; CV$\approx0.3$--0.5 gives smaller, model-dependent gains (App.~\ref{app:cv_fs_decomp}). \textbf{Tuning.} In short profiling runs, sweep $L_{\max}$ upward from 2--4$\times$ mean length and choose the smallest stable setting in the near-fastest emitted-sample throughput band; keep the default input depth unless a small \texttt{pf}/$D$ sweep gives a clear throughput gain (Table~\ref{tab:ablation_D}; App.~\ref{app:clamp}). After speed-first profiling, run full-epoch quality/stability validation; validation and benchmark scores are not HP-selection inputs. Use default join mode for final training runs; App.~\ref{app:join_tradeoff} shows that its strict identity-coverage contract has negligible throughput cost. The relaxed non-join termination is retained only for constrained runtime integrations that cannot support the join-style drain-before-finish protocol. \textbf{Method choice.} Use packing when text-only varlen attention/boundaries are wired; use ODB for multimodal or augmentation-heavy stacks with online lengths and costly cache/rewrite upkeep.

\section{Related Work}
\label{sec:related}
\vspace{-4pt}
\paragraph{Packing and length-aware batching.} Sequence packing removes padding by concatenation/masking~\cite{krell2021efficient}. In transformer stacks, efficient contamination-free packing needs boundary-aware masks/loss handling and often varlen attention/kernel support~\cite{dao2022flashattention,dao2023flashattention2} plus framework boundary plumbing; in our Qwen-VL/DeepSpeed/LLaMA-Factory stack it is a model-side intervention, not a pure DataLoader swap. Token-budget batchers in fairseq/NeMo/OpenNMT/Composer~\cite{ott2019fairseq,kuchaiev2019nemo,klein2017opennmt,mosaicml2022composer} and GMT-/BMT-oracles build offline length caches that rebuild under augmentation/template/cutoff changes. PyTorch samplers~\cite{paszke2019pytorch} and HuggingFace length grouping~\cite{wolf2020transformers} do not solve runtime variable group-count alignment across DDP ranks. ODB occupies the runtime DDP-aware DataLoader slot: no offline cache or model rewrite.
\paragraph{Complementary training-system techniques.} Sequence/context parallelism~\cite{korthikanti2022sequenceparallel,li2023sequenceparallelism,liu2023ringattention} shards one long sample across ranks; ODB batches variable-length samples. Inference continuous batching~\cite{yu2022orca,kwon2023efficient} lacks DDP step-count constraints. DeepSpeed Data Efficiency~\cite{li2022deepspeed} handles sampling/curriculum routing. ODB operates at the DataLoader boundary and composes with training-stack mechanisms such as DeepSpeed~\cite{rasley2020deepspeed}, DDP~\cite{li2020pytorch}, gradient accumulation, LoRA-style adapters~\cite{hu2022lora}, and LLaMA-Factory~\cite{zheng2024llamafactory}.

\section{Limitations}
\label{sec:limitations}
ODB's clamped memory rule $B(l){=}\max(\lfloor L_{\max}/l\rfloor,1)$ is a first-order activation-memory proxy, so $L_{\max}$ must be swept per model, optimizer, precision, and attention stack. Gains are data-dependent: near-uniform/low-CV workloads leave less long-tail slack, and we do not isolate an iso-token or matched-update causal ablation; reported speedups are batch-system operating points combining padding reduction, useful-batch growth, and input overlap. Empirical claims cover multimodal fine-tuning and MM-Mix, mainly single-node H20 runs, with two-node validation in App.~\ref{app:multinode}. The metadata alignment exchange should be revalidated at larger world sizes or heterogeneous nodes. ZeRO-3 and FSDP remain outside the evaluation scope; gradient accumulation follows the same aligned micro-step schedule and is supported by isolated validation runs. Default join mode gives strict identity coverage with negligible measured cost; opt-in non-join relaxes termination to cumulative sample-quota closure (App.~\ref{app:join_tradeoff}). ODB complements, rather than replaces, model-side packing or varlen-attention systems when those are available.

\section{Conclusion}
\label{sec:conclusion}
We introduced \odb, an online dynamic batcher that observes realized lengths after runtime preprocessing, tokenization, and multimodal expansion, then forms DDP-safe variable-size batches at the DataLoader boundary without changing the model, optimizer, attention kernel, or dataset. Max-Based Bidirectional Group Alignment provides strict identity coverage under default join mode, sample-quota closure under opt-in non-join mode, and deadlock-free bounded termination. Empirically, ODB delivers 1.58--2.51$\times$ throughput across public 2B/8B Full FT and LoRA rows and 4.43$\times$ on production MM-Mix with Standard-comparable validation and benchmark metrics. On 8B Full FT it reaches $1.77\times$ on UltraChat, $1.73\times$ on LLaVA, and $2.46\times$ on ShareGPT4o; under LoRA it reaches up to $2.51\times$. These results narrow the gap between fixed-batch training and stronger offline or model-side batching methods while avoiding scalar length caches and model-side packing rewrites.

\bibliographystyle{plain}
\bibliography{odb_paper}

\appendix

\section{Cross-Rank Group Alignment Protocol}
\label{app:cross_rank_protocol}

This appendix gives the full algorithm and supporting details summarized in Section~\ref{sec:alignment}.
Let $\mathcal{G}_r$ be rank $r$'s current candidate-group list, $G_r=|\mathcal{G}_r|$, and $\mathcal{A}=\{r:G_r>0\}$ the active-rank set in a protocol round.
When $\mathcal{A}$ is nonempty, ODB's alignment target is computed only over active ranks:
\begin{equation}
    T_{\mathrm{grp}} = \max\!\left(\min\!\left(\max_{r\in\mathcal{A}} G_r, \; C_{\min}^{+}, \; S_{\min}^{+}\right), \; 1\right)
    \label{eq:target}
\end{equation}
where $C_r$ is rank $r$'s output-slot capacity, $S_r$ is its buffered-sample count, and $C_{\min}^{+}=\min_{r\in\mathcal{A},\,C_r>0}C_r$ and $S_{\min}^{+}=\min_{r\in\mathcal{A},\,S_r>0}S_r$ are the positive minima over active ranks.
Excluding zero-capacity/zero-sample ranks prevents an empty rank from collapsing the target to zero.

After computing $T_{\mathrm{grp}}$, each active rank adjusts:
\emph{Split} (upward, $G_r < T_{\mathrm{grp}}$): scanning groups in reverse order, the first group with $\geq 2$ samples is found and its last sample is extracted to form a new singleton; repeat until $G_r{=}T_{\mathrm{grp}}$.
\emph{Overflow} (downward, $G_r > T_{\mathrm{grp}}$): the $T_{\mathrm{grp}}$ largest groups are retained and samples from removed groups are returned to the buffer for reuse.
This \emph{overflow recirculation} ensures no samples are permanently discarded.

\begin{algorithm}[ht]
\caption{Max-Based Bidirectional Group Alignment}
\label{alg:alignment}
\begin{algorithmic}[1]
\REQUIRE Candidate group lists $\{\mathcal{G}_0, \ldots, \mathcal{G}_{W-1}\}$, active set $\mathcal{A}=\{r:|\mathcal{G}_r|>0\}$, $C_{\min}^{+}$, $S_{\min}^{+}$ (positive values only)
\ENSURE All active ranks have exactly $T_{\mathrm{grp}}$ groups; inactive ranks remain idle; overflow samples returned to buffer
\STATE $T_{\mathrm{grp}} \gets \max(\min(\max_{r\in\mathcal{A}} |\mathcal{G}_r|, C_{\min}^{+}, S_{\min}^{+}), 1)$
\FOR{each active rank $r\in\mathcal{A}$}
    \IF{$|\mathcal{G}_r| < T_{\mathrm{grp}}$}
        \WHILE{$|\mathcal{G}_r| < T_{\mathrm{grp}}$ \textbf{and} $\exists$ group with $|g| \geq 2$}
            \STATE Scan from last group backward; find first $g^*$ with $|g^*| \geq 2$; extract its last sample as a new singleton group
        \ENDWHILE
    \ELSIF{$|\mathcal{G}_r| > T_{\mathrm{grp}}$}
        \STATE Sort groups by size (descending); keep top-$T_{\mathrm{grp}}$; return remaining samples to data buffer
    \ENDIF
\ENDFOR
\end{algorithmic}
\end{algorithm}

\paragraph{Communication overhead.} The primary metadata round performs one \texttt{all\_gather} of $(2 + 2\,\text{buffer\_size}) \cdot W \cdot$ \texttt{sizeof(int64)} bytes ($\approx$128\,KB at $W{=}8$, \texttt{buffer}=1024). Exact token-level loss scaling can trigger a second token-count gather under the deterministic all-rank predicate described in Section~\ref{sec:alignment}. The Gloo channel runs on CPU and overlaps with GPU compute.

\section{Loss Scaling Derivation}
\label{app:loss_scaling}

Let rank $r \in \{0,\ldots,W{-}1\}$ hold $n_r$ samples with $t_r$ valid tokens and per-token losses $\ell_{r,i,k}$.
Write $N{=}\sum_r n_r$, $T_{\mathrm{tok}}{=}\sum_r t_r$.
The reference loss (single-rank pass under per-token mean reduction) is
\begin{equation}
\mathcal{L}^{\star} = \frac{1}{T_{\mathrm{tok}}} \sum_{r,i,k} \ell_{r,i,k}.
\label{eq:target_loss}
\end{equation}
Each rank locally computes $\bar{\mathcal{L}}_r = (1/t_r) \sum_{i,k} \ell_{r,i,k}$; naive DDP yields $\tfrac{1}{W} \sum_r \bar{\mathcal{L}}_r$, which equals $\mathcal{L}^{\star}$ only in the degenerate case $t_r \equiv T_{\mathrm{tok}}/W$.
For weights $\{w_r\}$ summing to one, replacing $\bar{\mathcal{L}}_r$ with $\bar{\mathcal{L}}_r \cdot w_r \cdot W$ makes DDP's post-averaging output equal $\sum_r w_r \bar{\mathcal{L}}_r$; expanding shows the unique exact choice is $w_r = t_r/T_{\mathrm{tok}}$.
Sample-level weighting ($w_r = n_r/N$) is exact only when the average tokens per sample $t_r/n_r$ is identical across ranks.

\paragraph{Approximate vs.\ exact mode.}
Eq.~\ref{eq:exactness} requires per-group token counts consistent with the \emph{post}-alignment grouping.
The primary \texttt{all\_gather} piggybacks pre-alignment counts without an extra communication round; in \emph{approximate} mode adjusted counts are estimated from $\bar\tau_r{=}t_r^{\text{orig}}/n_r^{\text{orig}}$.
Exact mode uses the primary counts when alignment is a no-op and otherwise triggers a deterministic second \texttt{all\_gather} to re-broadcast post-alignment counts, preserving deadlock-freedom.
Empirically, in the ShareGPT4o profiling sweep, approximate and exact token-level modes have unscaled losses within 0.004, but only exact mode satisfies Eq.~\ref{eq:exactness} bit-precisely; all main experiments use exact token-level scaling.

\section{Formal Proofs}
\label{app:proofs}

This appendix gives explicit, full proofs of Theorems~\ref{thm:zero-discard} and~\ref{thm:deadlock-free} using (i) a per-rank state machine that tracks every place a sampler view may reside, and (ii) an explicit Lyapunov potential function that strictly decreases on emission rounds while skip rounds are finite. Sampler views are conserved as multiset membership at all times, so ``no leak'' is by inspection of the transition rules; identity coverage is handled by the view-to-identity projection.

\subsection{Per-Rank State and Transition Rules}
\label{app:proofs_state}

Fix world size $W$, rank index $r \in \{0, \ldots, W-1\}$, and dataset identity set $\mathcal{I} = \{0, \ldots, N-1\}$. \texttt{DistributedSampler(drop\_last=False)} produces a per-rank sampler-view sequence $\mathcal{D}_r$ of size $|\mathcal{D}_r| = \lceil N/W \rceil$ after padding the global shuffled index list to $M := W \cdot \lceil N/W \rceil$ views and stride-sharding it across ranks. View positions are disjoint across ranks; their identity projection covers $\mathcal{I}$, with $P := M - N$ deterministic tail-padding views that cyclically re-use boundary identities to make per-rank counts equal. We track the view multiset $\mathcal{D}_r$ (so an identity that appears as a padding view is counted as a distinct sampler element); identity-level coverage is the cardinality of the corresponding identity set. At protocol round $k$, rank $r$'s state is
\[
s_r^{(k)} \;=\; \bigl(R_r^{(k)},\; Q_r^{(k)},\; B_r^{(k)},\; E_r^{(k)}\bigr),
\]
where the four \emph{pairwise disjoint} components partition $\mathcal{D}_r$:
\begin{itemize}[nosep,leftmargin=1.2em]
\item $R_r^{(k)}$ --- \textbf{sampler-pending}: views $\mathcal{D}_r$ has not yet yielded.
\item $Q_r^{(k)}$ --- \textbf{worker queue}: views in flight from worker subprocesses to the collate process.
\item $B_r^{(k)}$ --- \textbf{collate buffer}: views received by collate but not yet emitted to the trainer.
\item $E_r^{(k)}$ --- \textbf{emitted}: views already delivered to the trainer in some prior batch.
\end{itemize}
The protocol exposes three transition primitives, each of which moves sampler views between two components without creation or destruction:
\begin{align*}
\textsc{Fetch}_r(F):\;\; & F \subseteq R_r^{(k)},\; \delta=|F|;\quad
R_r^{(k+1)} = R_r^{(k)} \setminus F,\; Q_r^{(k+1)} = Q_r^{(k)} \uplus F\\
\textsc{Drain}_r:\;\; & Q_r^{(k)} \to B_r^{(k)} \\
\textsc{Emit}_r(g):\;\; & B_r^{(k)} \to E_r^{(k)},\; \text{moves group } g \subseteq B_r^{(k)}
\end{align*}
Let $D := \max(\texttt{prefetch\_factor} \cdot \texttt{num\_workers}, \texttt{buffer\_size})$ denote the configured per-rank outstanding-depth envelope. The iterator schedules fetch/drain so that the fetched-but-not-emitted view set $Q_r \uplus B_r$ stays within this envelope; $B_r$ contains the current collate buffer plus any overflow groups recirculated after alignment. No transition deletes elements: every transition is one of the three above or the local-rank no-op of a \texttt{skip\_output} round.

\begin{lemma}[No-Leak Invariant]
\label{lem:no_leak}
At every round $k$ and on every rank $r$,
\[
R_r^{(k)} \uplus Q_r^{(k)} \uplus B_r^{(k)} \uplus E_r^{(k)} \;=\; \mathcal{D}_r.
\]
\end{lemma}
\begin{proof}
Base case $k=0$: $R_r^{(0)}=\mathcal{D}_r$, $Q_r^{(0)}=B_r^{(0)}=E_r^{(0)}=\emptyset$. Inductive step: each transition primitive moves a (possibly empty) subset between two components on the left-hand side, leaving the disjoint union invariant.
\end{proof}

\subsection{Lyapunov Potential and Bounded Termination}
\label{app:proofs_lyapunov}

Define the global potential
\[
\Phi^{(k)} \;:=\; \sum_{r=0}^{W-1} \bigl(|R_r^{(k)}| + |Q_r^{(k)}| + |B_r^{(k)}|\bigr) \;=\; M - \sum_{r} |E_r^{(k)}|,
\]
where $M = W \cdot \lceil N/W \rceil$ is the total sampler view count (App.~\ref{app:proofs_state}). $\Phi^{(k)} \in \{0, 1, \ldots, M\}$ is a non-negative integer. We track its evolution.

\begin{lemma}[Bounded-round progress]
\label{lem:progress}
On any active outer round $k$ (some rank has $|B_r^{(k)}| > 0$), let $G_{\mathrm{cur}}^{(k)}=\max_{r:\,G_r>0}G_r$, and let $C_{\min}^{(k),+}$ and $S_{\min}^{(k),+}$ be the round-$k$ positive output-slot and buffered-sample minima over active ranks. Define $t_k = \max(1, \min(G_{\mathrm{cur}}^{(k)}, C_{\min}^{(k),+}, S_{\min}^{(k),+}))$. Local alignment may split groups before emission, but it does not move views into $E_r$ and performs finitely many splits, bounded by the number of buffered views/groups in $B_r^{(k)}$ because each split only isolates existing views. After local alignment, the outer round guarantees one of the following:
\begin{enumerate}[label=(\alph*),nosep,leftmargin=1.6em,topsep=1pt]
\item (\textup{\textsc{Emit}} round.) $\Phi^{(k+1)} \le \Phi^{(k)} - 1$.
\item (\textup{\textsc{Skip}} round.) $\Phi^{(k+1)} = \Phi^{(k)}$ and no view is emitted; the round only fetches/drains bounded in-flight views or observes sampler exhaustion.
\end{enumerate}
Skip rounds are finite because the fetched-but-not-emitted set is bounded by the outstanding-depth envelope $D$ and the sampler is finite; local split operations are internal to the outer round and do not affect the communication-round count.
\end{lemma}
\begin{proof}
Local split adjustment extracts views from existing groups but leaves them in $B_r$, so it preserves $\Phi$ and is bounded by the finite number of buffered views/groups. In a normal aligned case ($G_r=t_k$ for all active $r$), each active rank emits exactly $t_k$ groups containing $\sum_g |g| \ge t_k \ge 1$ views; hence $\Phi$ strictly decreases by at least one. In an overflow case ($G_r>t_k$), the top-$t_k$ groups are emitted and the remainder is recirculated into $B_r$; the emitted groups again contain at least one view, so $\Phi$ decreases. If no active rank has enough material to emit, the round is a \texttt{skip\_output} fetch/drain round; it cannot repeat indefinitely because outstanding in-flight views are bounded by $D$ and the sampler is finite.
\end{proof}

\begin{theorem}[Round-Count Bound]
\label{thm:bounded_term}
The protocol terminates in at most $\lceil N/W \rceil + O(D)$ rounds.
\end{theorem}
\begin{proof}
Let $q=\lceil N/W\rceil$ be the per-rank sampler-view quota under \texttt{DistributedSampler} with \texttt{drop\_last=False}. In every \textsc{Emit} outer round, aligned emission is not serialized by rank: after the bounded fetch/drain prefix, each unfinished active rank emits at least one sampler view in the shared step, while a rank with no remaining material is represented by the mode-specific finished state. Local split/overflow adjustment is internal to that round and does not add communication rounds. Therefore the shared protocol has at most $q$ emitting outer rounds before local sampler-view quotas are drained and the mode-specific termination predicate can be raised. Non-emitting \texttt{skip\_output} rounds only fetch/drain the bounded outstanding set ($|Q_r|+|B_r|\le D$ by the configured envelope) or observe sampler exhaustion, so they contribute an $O(D)$ epilogue. Thus each logical iteration terminates in at most $q+O(D)=\lceil N/W\rceil+O(D)$ outer rounds. This is a termination/quota bound only; identity-level coverage for non-join is not claimed by the proof and is handled empirically in App.~\ref{app:eta_identity}, while join mode gives the construction-level identity guarantee (Theorem~\ref{thm:strict-zero}).
\end{proof}

\subsection{Uniform all\_gather and Deadlock-Freedom}
\label{app:proofs_uniform}

\begin{lemma}[Uniform all\_gather invariant]
\label{lem:uniform_gather}
Every rank executes one unconditional primary \texttt{all\_gather} per outer iteration. If exact token-level loss scaling needs post-alignment token counts, the optional second \texttt{all\_gather} (Appendix~\ref{app:loss_scaling}) is gated by a deterministic predicate $\phi(\texttt{all\_n\_groups}, \texttt{all\_idx\_budgets})$ computed identically on every rank, so it is either executed by all or by none.
\end{lemma}
\begin{proof}
The protocol executes the primary \texttt{all\_gather} before any branch can return. The optional second \texttt{all\_gather}'s entry predicate is a pure function of the same broadcast tensor, so all ranks evaluate it identically.
\end{proof}

\begin{proof}[Proof of Theorem~\ref{thm:deadlock-free} (Bounded Termination $+$ Deadlock-Free)]
By induction on iteration $i$: if all ranks enter iteration $i$ together (base $i{=}0$ trivial), Lemma~\ref{lem:uniform_gather} ensures all execute the same \texttt{all\_gather} call(s) and observe identical broadcast tensors. The mode-specific termination predicate---all ranks advertise local finish in default join mode, or any rank advertises \texttt{n\_groups$_r$}$=-1$ in opt-in non-join---is computed from those shared tensors and therefore evaluated identically; hence all enter iteration $i{+}1$ together or all exit together. No rank can block on \texttt{all\_gather} because all ranks reach it. Bounded termination follows from Theorem~\ref{thm:bounded_term}.
\end{proof}

\subsection{Sample-Quota Closure}
\label{app:proofs_quota}

\begin{proof}[Proof of Theorem~\ref{thm:zero-discard} (No-Leak $+$ Sample-Quota Closure)]
The no-leak claim is exactly Lemma~\ref{lem:no_leak}.

\emph{Sample-quota closure.} The trainer side maintains an emitted-sample counter, accumulates the realized global per-step emitted-sample count, and terminates once the cumulative emit count reaches $N$. When a logical \texttt{DistributedSampler} iteration ends early (Theorem~\ref{thm:deadlock-free}), the outer training loop starts the next logical iteration with a re-shuffled sampler; the stopping condition is unaffected by these chained iterations. Let $S_{\max}$ be the largest realized global emit count of one aligned trainer step under the configured outstanding-depth envelope; the final quota crossing occurs in one such step, so $S_{\text{emit}} - N \le S_{\max}$.
\end{proof}

\begin{lemma}[Logical-discard bound, non-join]
\label{lem:logical_bound}
Within one logical sampler iteration under non-join termination, let $U_r:=Q_r \uplus B_r$ be the fetched-but-not-emitted outstanding set on rank $r$; sampler-pending views $R_r$ are not fetched and are not counted in $U_r$. Then the configured outstanding-depth envelope gives $|U_r| \le D$, hence $\eta_{\text{logical}} := \tfrac{1}{N} \sum_r |U_r| \le W \cdot D / N$. This is a per-iteration envelope, not a terminal identity-coverage statement. Non-join provides cumulative sample-quota closure (Theorem~\ref{thm:zero-discard}), and Appendix~\ref{app:eta_identity} reports terminal identity coverage for the evaluated Full FT cells. Strict per-iteration logical zero-discard and identity coverage are obtained by join mode (Theorem~\ref{thm:strict-zero}), where $|U_r|=0$ at termination by drain-then-signal.
\end{lemma}
\begin{proof}
At a non-join stop point, fetched-but-not-emitted views can reside only in the worker queue $Q_r$ or collate buffer $B_r$. ODB's outstanding-depth envelope bounds this set by $D$ on every rank, including overflow groups recirculated into $B_r$ after alignment. Therefore $|U_r|=|Q_r|+|B_r|\le D$ for every rank at the non-join stop point. Summing over $W$ ranks gives the stated bound.
\end{proof}

\subsection{Empirical \texorpdfstring{$\eta_{\text{quota}}$}{eta-quota}}
\label{app:eta_measurement}

We instantiate Theorem~\ref{thm:zero-discard} on the audited public Full-FT ODB runs by computing $\eta_{\text{quota,emp}} := \max(0, 1 - S_{\text{emit}}/N)$ from terminal trainer state. Across the 18 audited runs (3 datasets $\times$ 2 ODB configurations $\times$ 3 seeds), $\eta_{\text{quota,emp}} = 0$ uniformly with terminal epoch in $\{1.0000000, 1.0000334, 1.0000735\}$; the bounded overshoot is the final batch crossing the quota threshold (Theorem~\ref{thm:zero-discard}). The same quota check on the six 1000-sample synthetic distributions listed in App.~\ref{app:setup_details} gives $\eta_{\text{quota,emp}}=0$ with terminal epoch in $\{1.0000,1.0001\}$. Per-iteration logical bounds $\eta_{\text{logical}} \le W \cdot D / N$ for representative configurations are reported in Table~\ref{tab:eta_logical}; we use them only as worst-case protocol envelopes, while terminal measured quota and identity metrics are reported separately.

\begin{table}[ht]
\centering
\small
\begin{tabular}{lrrrr}
\toprule
Configuration & $N$ & $W$ & $D$ & $\eta_{\text{logical}}$ bound \\
\midrule
LLaVA 8B ($D{=}4096$)                       & 157,712 & 8 & 4,096 & 20.8\% \\
UltraChat 8B (ml8k pf256 buf256)            & 207,865 & 8 & 1,024 & 3.9\% \\
UltraChat 8B (ml8k pf1024 buf1024)          & 207,865 & 8 & 4,096 & 15.8\% \\
UltraChat 8B (ml16k pf512 buf1024)          & 207,865 & 8 & 2,048 & 7.9\% \\
ShareGPT4o 8B (ml4k pf1024)                 &  54,424 & 8 & 4,096 & 60.2\% \\
MM-Mix 8B (ml8k pf256)                      & 545,178 & 8 & 1,024 & 1.5\% \\
MM-Mix 8B (extreme, ml4k pf2048)            & 545,178 & 8 & 8,192 & 12.0\% \\
\bottomrule
\end{tabular}
\caption{Per-iteration logical-discard upper bound $\eta_{\text{logical}} \le W \cdot D / N$ (Lemma~\ref{lem:logical_bound}). The bound is a worst-case envelope on per-iteration un-emitted sampler views; cumulative-count $\eta_{\text{quota}}$ is driven to 0 by the trainer-side emitted-sample counter (Theorem~\ref{thm:zero-discard}) in the audited runs, and per-sample $\eta_{\text{identity}}$ is empirically $0$ on four terminal-state Full FT cells (Ultra and SGPT, 2B and 8B; App.~\ref{app:eta_identity}, Table~\ref{tab:eta_identity_audit}; surplus emits matching \texttt{DistributedSampler}'s deterministic tail-padding count, with the same 2B/8B surplus). Strict per-iteration $\eta_{\text{logical}}=0$ with identity coverage \emph{by construction} (independent of straggler asymmetry) is reserved for join mode (Theorem~\ref{thm:strict-zero}).}
\label{tab:eta_logical}
\end{table}

\subsection{Terminal Identity Coverage}
\label{app:eta_identity}

The $\eta_{\text{logical}}$ bound of Table~\ref{tab:eta_logical} is a worst-case per-iteration envelope. We also report the terminal identity metric $\eta_{\text{identity}} := 1 - |\bigcup_r \mathrm{IDs}_r| / N$, the fraction of dataset identities not emitted by the union of ranks at training termination. We measure it on four full-epoch Full FT cells: Ultra and SGPT at both 2B and 8B. These cells use the same protocol-relevant settings as their \S\ref{sec:eval} counterparts---world size $W{=}8$, ODB knobs, seed, batch policy, and \texttt{DistributedSampler(drop\_last=False)}---while using the full unsplit datasets for unambiguous $N$ accounting.

\paragraph{Terminal-state result.}
\textbf{All four measured cells have $\eta_{\text{identity}} = 0$}: $|\bigcup_r \mathrm{IDs}_r| = N$. The surplus emit count $\sum_r |\text{emits}_r| - N \in \{4, 7\}$ matches the deterministic tail padding of \texttt{DistributedSampler} with \texttt{drop\_last=False}: $W - (57{,}284 \bmod 8) = 4$ for SGPT and $W - (207{,}865 \bmod 8) = 7$ for Ultra. Matched 2B and 8B runs have the same surplus count on each dataset (Table~\ref{tab:eta_identity_audit}).

\paragraph{Measured invariants.}
Two empirical invariants, both observed in all four measured cells, are sufficient to conclude $\eta_{\text{identity}} = 0$:

\noindent\textbf{Shard-bounded emit, no extra duplicates.} Every emitted ID is a valid dataset identity (a member of $r$'s own \texttt{DistributedSampler} shard $S_r$), and the union of per-rank emitted-ID records contains no cross-rank duplicates beyond the $P = W \lceil N/W \rceil - N$ deterministic sampler-padding reuses (4 for SGPT, 7 for Ultra).

\noindent\textbf{Per-rank emit count.} By termination each rank has emitted exactly $q = \lceil N/W \rceil$ sampler views (7,161 each for SGPT; 25,984 each for Ultra; same in matched 2B/8B runs).

These two invariants then yield identity coverage as a deterministic implication:

\begin{proposition}[Identity closure from measured invariants]
\label{prop:identity_closure_audit}
Consider a non-join ODB run with \texttt{DistributedSampler(drop\_last=False)} over $N$ dataset identities and world size $W$. Let $q = \lceil N/W \rceil$, $M = Wq$, and $P = M - N$ be the deterministic sampler-padding surplus. If the terminal emitted-ID logs satisfy:
\emph{(a)} all emitted IDs are valid dataset identities and contain no duplicates beyond the $P$ deterministic sampler-padding reuses; and
\emph{(b)} each rank emits exactly $q$ sampler views by termination,
then the terminal identity coverage is exact: $\eta_{\text{identity}} = 0$.
\end{proposition}

\begin{proof}
By condition (b), the run emits $M = Wq = N + P$ sampler views in total. By condition (a), the number of duplicate identity emissions is at most the deterministic padding surplus $P$. Therefore the number of unique emitted dataset identities is at least $M - P = N$. Since all emitted IDs are valid dataset identities (and the dataset has only $N$ identities), the unique-emit set is a subset of size at least $N$ of an $N$-element set, hence equals it: $\bigcup_r \mathrm{IDs}_r = \{0, \ldots, N-1\}$, and $\eta_{\text{identity}} = 1 - |\bigcup_r \mathrm{IDs}_r| / N = 0$.
\end{proof}

\noindent The argument is conditional: it converts the two measured invariants into a closed-form identity-closure conclusion for these cells, but does not promote the empirical premises into a non-join formal guarantee in general (which is reserved for join mode, Theorem~\ref{thm:strict-zero}).

\begin{table}[ht]
\centering
\small
\setlength{\tabcolsep}{3.5pt}
\begin{tabular}{llrrrrrr}
\toprule
\multicolumn{1}{l}{Measured cell} & ODB setting & $N$ & $D$ & emits / rank & total emits & dup vs $W{-}N\!\bmod\!W$ & $\eta_{\text{identity}}$ \\
\midrule
\multicolumn{8}{l}{\emph{Terminal-state Full FT cells (full epoch)}} \\
SGPT 2B  & ml4k pf1024 & 57{,}284 & 4{,}096 & 7{,}161 ($\times 8$) & 57{,}288 & 4 vs $\mathbf{4}$ \checkmark & $\mathbf{0\%}$ \\
SGPT 8B  & ml4k pf1024 & 57{,}284 & 4{,}096 & 7{,}161 ($\times 8$) & 57{,}288 & 4 vs $\mathbf{4}$ \checkmark & $\mathbf{0\%}$ \\
Ultra 2B & ml8k buf256 & 207{,}865 & 1{,}024 & 25{,}984 ($\times 8$) & 207{,}872 & 7 vs $\mathbf{7}$ \checkmark & $\mathbf{0\%}$ \\
Ultra 8B & ml8k buf256 & 207{,}865 & 1{,}024 & 25{,}984 ($\times 8$) & 207{,}872 & 7 vs $\mathbf{7}$ \checkmark & $\mathbf{0\%}$ \\
\bottomrule
\end{tabular}
\caption{Terminal identity coverage on four Full FT cells (Ultra and SGPT, 2B and 8B, full epoch). The ``ODB setting'' column names the two audited configurations; ``dup vs $W{-}N\!\bmod\!W$'' compares observed surplus emits with the deterministic \texttt{DistributedSampler(drop\_last=False)} tail-padding count. All rows cover $N$ unique dataset identities; matched 2B/8B runs have the same surplus count on each dataset.}
\label{tab:eta_identity_audit}
\end{table}

\paragraph{Relation to the $\eta_{\text{logical}}$ envelope.}
The distinction is that $\eta_{\text{logical}}$ bounds sampler views that may remain un-emitted within one logical iteration, whereas $\eta_{\text{identity}}$ measures the union of emitted dataset identities at termination. In the four measured Full FT cells, the union covers all $N$ identities. The construction-level identity guarantee for arbitrary configurations is provided by join mode (Theorem~\ref{thm:strict-zero}).

\section{Grouping Algorithm Example}
\label{app:grouping_example}

We illustrate the grouping algorithm (Section~\ref{sec:batch_sizing}) with a concrete example.
Suppose $L_{\max} = 1000$ and a rank's buffer contains four samples with lengths $\{100, 200, 500, 800\}$.

\medskip
\noindent Sort ascending, $[100, 200, 500, 800]$, initialize threshold $t = 1$, and iterate from longest to shortest:

\begin{enumerate}[nosep]
    \item \textbf{Sample 800}: group $= [800]$. Size $= 1$, $t = 1$; size $\geq t$, so finalize $G_1 = [800]$.\\
    Update $t \leftarrow B(800) = \lfloor 1000/800 \rfloor = 1$.
    \item \textbf{Sample 500}: group $= [500]$. Size $= 1$, $t = 1$; size $\geq t$, so finalize $G_2 = [500]$.\\
    Update $t \leftarrow B(500) = \lfloor 1000/500 \rfloor = 2$.
    \item \textbf{Sample 200}: group $= [200]$. Size $= 1$, $t = 2$; size $< t$, continue.
    \item \textbf{Sample 100}: group $= [100, 200]$. Size $= 2$, $t = 2$; size $\geq t$, so finalize $G_3 = [100, 200]$.\\
    Update $t \leftarrow B(100) = \lfloor 1000/100 \rfloor = 10$.
\end{enumerate}

\noindent\textbf{Result}: three groups, ordered from short to long:

\begin{center}
\small
\begin{tabular}{lccc}
\toprule
Group & Samples (lengths) & Padded to & Padded tokens \\
\midrule
$G_3$ & 100, 200 & 200 & 400 \\
$G_2$ & 500 & 500 & 500 \\
$G_1$ & 800 & 800 & 800 \\
\bottomrule
\end{tabular}
\end{center}

\noindent The threshold carry-over is the key mechanism: $G_2$ contains only one sample because the threshold inherited from $G_1$ is $t = B(800) = 1$, which is already met.
Meanwhile, the updated $t = B(500) = 2$ requires $G_3$ to accumulate two samples before finalizing, grouping the two shortest sequences together.
With more samples of similar lengths (the typical case), each group's padded-token cost approaches~$L_{\max}$.

\section{Protocol State Machine}
\label{app:state_machines}

Figure~\ref{fig:sm_nojoin} renders the non-join per-iteration state machine. It is the executable transition system abstracted in Appendix~\ref{app:proofs}: the Lyapunov potential of Appendix~\ref{app:proofs_lyapunov} contracts on active-emission rounds after bounded local split/overflow adjustment, and is unchanged on finite \texttt{skip} rounds (Lemma~\ref{lem:progress}). The transition to \texttt{Logical Stop} occurs when any rank signals $g_r=-1$; the trainer-side quota counter then chains logical iterations until the sample quota is met (Theorem~\ref{thm:zero-discard}). In the figure, $g_r$ denotes rank~$r$'s gathered group-status signal and $b_r$ denotes rank~$r$'s \texttt{idx\_budget}.

\begin{figure}[ht]
    \centering
    \begin{tikzpicture}[
        state/.style={circle, draw, minimum size=1.1cm, font=\scriptsize, align=center, line width=0.6pt},
        init/.style={circle, fill=black, minimum size=4pt, inner sep=0pt},
        arr/.style={-{Stealth[length=4pt]}, line width=0.6pt},
        trans/.style={font=\tiny, align=center, text=gray!30!black},
    ]
    \node[init] (start) {};
    \node[state, fill=green!10, right=1.4cm of start] (run) {Running};
    \node[state, fill=red!10, right=3.0cm of run] (done) {Logical\\Stop};
    \node[init, right=1.4cm of done] (end) {};
    \node[font=\tiny, right=0.1cm of end] {exit};

    \draw[arr] (start) -- (run);
    \draw[arr] (run) to[loop above, looseness=5, out=110, in=70] node[trans, above=2pt] {active emit:\\$\forall r: g_r \geq 0$\\$\min_r b_r > 0$\\$\Phi^{(k+1)} \le \Phi^{(k)} - 1$} (run);
    \draw[arr] (run) to[loop below, looseness=5, out=-110, in=-70] node[trans, below=2pt] {skip\_output:\\$\forall r: g_r \geq 0$\\$\min_r b_r = 0$\\$\Phi$ unchanged, $|R|$ drains} (run);
    \draw[arr] (run) -- node[trans, above] {$\exists r: g_r = -1$} (done);
    \draw[arr] (done) -- node[trans, above] {SENTINEL} (end);
    \end{tikzpicture}
    \caption{Per-iteration state machine of the Unified Loop Protocol in non-join mode. Transition guards reference the per-rank state of Appendix~\ref{app:proofs_state}; $g_r$ is rank~$r$'s gathered group-status signal, $b_r$ is rank~$r$'s \texttt{idx\_budget}, and $\Phi$ is the Lyapunov potential of Appendix~\ref{app:proofs_lyapunov}. The active self-loop contracts $\Phi$ on emission after bounded local split/overflow adjustment (Lemma~\ref{lem:progress}); the \texttt{skip\_output} self-loop is finite by sampler exhaustion.}
    \label{fig:sm_nojoin}
\end{figure}
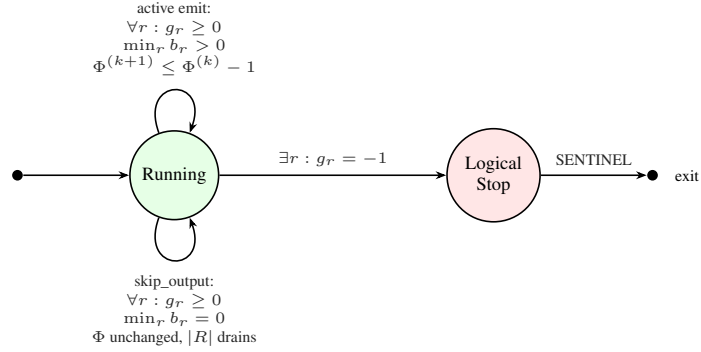

\section{Two-Node Validation}
\label{app:multinode}

\paragraph{Empty-rank join-mode audit.}
We evaluate a two-node H20 empty-rank multi-node case with 8 GPUs per node (16 ranks) outside the equal-samples premise of Theorem~\ref{thm:zero-discard}. Because this setting intentionally violates that premise, we use it only as a liveness audit for deadlock-freedom and bounded termination, not as a quota- or identity-coverage audit.

We run the DataLoader-level ODB path in join mode. The one-node 8-rank companion run terminates with \texttt{PASS world=8}; the 2-node run uses 8 H20 GPUs per node (16 ranks total), sets global rank 15 as the exhausted empty rank, and terminates with \texttt{PASS world=16 empty\_rank=15}. Table~\ref{tab:multinode_audit} reports the per-rank sampler assignment, trainer-side outputs, and post-join liveness flags. In both audits all active ranks emit batches, the exhausted empty rank emits zero batches, and every collate subprocess exits.

\begin{table}[ht]
\centering
\caption{Two-node empty-rank join-mode audit: per-rank sampler views and trainer-side outputs. Steps are optimizer updates per rank; Emitted is the number of sampler views emitted to the trainer, so Emitted can exceed Steps under dynamic batching. Done=1 indicates clean collate-subprocess exit after join.}
\label{tab:multinode_audit}
\scriptsize
\setlength{\tabcolsep}{2.2pt}
\begin{tabular}{@{}l*{16}{r}@{}}
\toprule
Rank & 0 & 1 & 2 & 3 & 4 & 5 & 6 & 7 & 8 & 9 & 10 & 11 & 12 & 13 & 14 & 15 \\
\midrule
Assigned & 47 & 46 & 45 & 44 & 43 & 42 & 41 & 40 & 39 & 38 & 37 & 36 & 35 & 34 & 33 & 0 \\
Steps & 35 & 35 & 33 & 33 & 33 & 33 & 31 & 31 & 31 & 31 & 29 & 29 & 29 & 29 & 27 & 0 \\
Emitted & 46 & 46 & 42 & 42 & 42 & 42 & 38 & 38 & 38 & 38 & 34 & 34 & 34 & 34 & 31 & 0 \\
Done & 1 & 1 & 1 & 1 & 1 & 1 & 1 & 1 & 1 & 1 & 1 & 1 & 1 & 1 & 1 & 1 \\
\bottomrule
\end{tabular}
\end{table}

The 21-view assigned--emitted difference is expected in this unequal-partition audit outside the theorem premise and is not interpreted as a coverage metric: the empty-rank construction breaks the equal per-rank sampler-quota premise of Theorem~\ref{thm:zero-discard}, and cross-rank alignment gates emission through the global target $T_{\mathrm{grp}}$ in Eq.~\ref{eq:target}. The measurement therefore targets liveness---no deadlock, bounded termination, and clean subprocess exit---rather than quota or identity closure.

These audits exercise the join-mode drain-then-signal path for exhausted empty ranks. They are liveness audits rather than the source of the main-table throughput rows, and they do not claim support for an active rank that advertises zero samples before it has reached the finished state.

\paragraph{Two-node Full FT training with MMMU-MC likelihood scores.}
We also report a two-node RDMA-enabled 16-rank Qwen3-VL Full FT validation. Multimodal Score cells use \texttt{MMMU-MC-choice-}\allowbreak\texttt{likelihood-v1}; training and validation-loss metrics use the same aggregation protocol as Table~\ref{tab:public_fullft}. All 72 two-node MMMU-MC seed cells pass the protocol validation checks (evaluated=850, excluded=50, total=900; no non-finite score audit failures). This table is multi-node validation evidence rather than the single-node headline setting of Table~\ref{tab:public_fullft}.
On 2B, ODB reaches $1.71\times$ on LLaVA, $3.76\times$ on ShareGPT4o, and $2.86\times$ on UltraChat; under MMMU-MC likelihood, ODB remains in the Standard/oracle quality band on both two-node 2B multimodal rows. On 8B, ODB reaches $1.80\times$ on LLaVA and $3.78\times$ on ShareGPT4o; it is highest on ShareGPT4o and within 0.08 pp of the best LLaVA score.

\begin{table}[ht]
\centering
\caption{\textbf{Two-node Full FT Results with MMMU-MC likelihood scores.} 2 nodes $\times$ 8 H20 GPUs (16 ranks), RDMA-enabled Qwen3-VL Full FT, DeepSpeed ZeRO-2/bf16. Rows are 3-seed mean$\pm$std. LLaVA and ShareGPT4o Score cells use \texttt{MMMU-MC-choice-likelihood-v1}; UltraChat uses MMLU. Oracle rows use scalar length caches for batch construction and exclude cache construction from reported training throughput.}
\label{tab:multinode_fft}
\tiny
\setlength{\tabcolsep}{1pt}
\resizebox{\linewidth}{!}{%
\begin{tabular}{@{}llcccc@{\hspace{4pt}}cccc@{}}
\toprule
 & & \multicolumn{4}{c}{\textbf{8B (2-node Full FT)}} & \multicolumn{4}{c}{\textbf{2B (2-node Full FT)}} \\
\cmidrule(lr){3-6} \cmidrule(lr){7-10}
Dataset & Method & sam/s & Speedup & Score & Val Loss & sam/s & Speedup & Score & Val Loss \\
\midrule
\multirow{7}{*}{\shortstack[l]{UltraChat\\CV=0.48\\MMLU}}
& Standard & $6.82 \pm 0.40$ & 1.00$\times$ & $74.73 \pm 0.17$\% & $0.856 \pm 0.001$ & $24.56 \pm 0.12$ & 1.00$\times$ & $58.82 \pm 0.30$\% & $0.971 \pm 0.001$ \\
& Sorted & $11.36 \pm 1.32$ & 1.67$\times$ & $75.14 \pm 0.17$\% & $0.883 \pm 0.002$ & $47.66 \pm 2.63$ & 1.94$\times$ & $59.04 \pm 0.07$\% & $1.018 \pm 0.001$ \\
& Packing\textsuperscript{a} & $20.19 \pm 0.66$ & 2.96$\times$ & $75.91 \pm 0.01$\% & $1.172 \pm 0.006$ & $72.23 \pm 1.44$ & 2.94$\times$ & $60.00 \pm 0.09$\% & $1.174 \pm 0.001$ \\
& GMT-oracle\textsuperscript{b} & $17.38 \pm 1.54$ & 2.55$\times$ & $75.10 \pm 0.19$\% & $0.856 \pm 0.001$ & $68.44 \pm 1.44$ & 2.79$\times$ & $59.28 \pm 0.10$\% & $0.992 \pm 0.001$ \\
& BMT-oracle\textsuperscript{b} & $17.92 \pm 0.99$ & 2.63$\times$ & $75.56 \pm 0.10$\% & $0.859 \pm 0.001$ & $64.05 \pm 3.02$ & 2.61$\times$ & $59.16 \pm 0.20$\% & $0.984 \pm 0.001$ \\
& HFG-oracle\textsuperscript{b} & $10.75 \pm 0.12$ & 1.58$\times$ & $75.19 \pm 0.17$\% & $0.857 \pm 0.001$ & $35.93 \pm 1.54$ & 1.46$\times$ & $59.05 \pm 0.34$\% & $0.980 \pm 0.001$ \\
& \textbf{ODB} & $\mathbf{18.80 \pm 0.03}$ & \textbf{2.76$\times$} & $75.47 \pm 0.10$\% & $0.859 \pm 0.001$ & $\mathbf{70.15 \pm 1.86}$ & \textbf{2.86$\times$} & $59.09 \pm 0.11$\% & $1.016 \pm 0.002$ \\
\midrule
\multirow{6}{*}{\shortstack[l]{LLaVA\\CV=0.29\\MMMU-MC}}
& Standard & $25.42 \pm 0.13$ & 1.00$\times$ & $54.63 \pm 0.07$\% & $1.032 \pm 0.001$ & $87.10 \pm 0.64$ & 1.00$\times$ & $40.08 \pm 0.56$\% & $1.172 \pm 0.002$ \\
& Sorted & $33.53 \pm 2.89$ & 1.32$\times$ & $54.82 \pm 0.61$\% & $1.131 \pm 0.002$ & $119.81 \pm 0.42$ & 1.38$\times$ & $38.63 \pm 0.14$\% & $1.294 \pm 0.001$ \\
& GMT-oracle\textsuperscript{b} & $47.98 \pm 0.56$ & 1.89$\times$ & $54.63 \pm 0.98$\% & $1.074 \pm 0.002$ & $155.91 \pm 2.91$ & 1.79$\times$ & $38.90 \pm 0.07$\% & $1.203 \pm 0.002$ \\
& BMT-oracle\textsuperscript{b} & $47.16 \pm 0.29$ & 1.86$\times$ & $54.51 \pm 0.18$\% & $1.074 \pm 0.002$ & $152.81 \pm 2.42$ & 1.75$\times$ & $39.21 \pm 0.34$\% & $1.203 \pm 0.002$ \\
& HFG-oracle\textsuperscript{b} & $33.75 \pm 0.33$ & 1.33$\times$ & $55.10 \pm 0.07$\% & $1.069 \pm 0.002$ & $113.67 \pm 0.51$ & 1.31$\times$ & $40.78 \pm 0.89$\% & $1.174 \pm 0.002$ \\
& \textbf{ODB} & $\mathbf{45.72 \pm 0.35}$ & \textbf{1.80$\times$} & $55.02 \pm 0.47$\% & $1.105 \pm 0.001$ & $\mathbf{148.80 \pm 0.78}$ & \textbf{1.71$\times$} & $40.24 \pm 0.83$\% & $1.239 \pm 0.002$ \\
\midrule
\multirow{6}{*}{\shortstack[l]{ShareGPT4o\\CV=1.00\\MMMU-MC}}
& Standard & $2.73 \pm 0.41$ & 1.00$\times$ & $53.06 \pm 0.85$\% & $1.218 \pm 0.006$ & $8.33 \pm 0.05$ & 1.00$\times$ & $41.53 \pm 0.20$\% & $1.293 \pm 0.006$ \\
& Sorted & $3.09 \pm 0.02$ & 1.13$\times$ & $53.29 \pm 0.12$\% & $1.240 \pm 0.006$ & $8.69 \pm 0.11$ & 1.04$\times$ & $39.84 \pm 0.25$\% & $1.320 \pm 0.006$ \\
& GMT-oracle\textsuperscript{b} & $3.37 \pm 0.02$ & 1.23$\times$ & $54.35 \pm 1.16$\% & $1.217 \pm 0.006$ & $8.50 \pm 0.01$ & 1.02$\times$ & $41.10 \pm 0.65$\% & $1.292 \pm 0.006$ \\
& BMT-oracle\textsuperscript{b} & $3.13 \pm 0.32$ & 1.15$\times$ & $53.61 \pm 0.74$\% & $1.218 \pm 0.006$ & $9.01 \pm 0.07$ & 1.08$\times$ & $41.18 \pm 1.14$\% & $1.293 \pm 0.006$ \\
& HFG-oracle\textsuperscript{b} & $2.91 \pm 0.35$ & 1.07$\times$ & $53.33 \pm 0.95$\% & $1.217 \pm 0.005$ & $9.59 \pm 0.08$ & 1.15$\times$ & $40.90 \pm 0.80$\% & $1.293 \pm 0.006$ \\
& \textbf{ODB} & $\mathbf{10.32 \pm 0.46}$ & \textbf{3.78$\times$} & $55.25 \pm 0.47$\% & $1.230 \pm 0.006$ & $\mathbf{31.35 \pm 0.37}$ & \textbf{3.76$\times$} & $41.45 \pm 0.67$\% & $1.356 \pm 0.005$ \\
\bottomrule
\end{tabular}}
\vspace{0.5mm}
\begin{minipage}{\linewidth}
\raggedright\scriptsize
\textsuperscript{a}Packing is UltraChat-only and not drop-in for multimodal training in our stack; bold throughput/speedup marks ODB among online DDP batchers. Packing val\_loss uses the packed-sequence denominator and is not directly comparable to per-original-sample rows.\quad
\textsuperscript{b}GMT-/BMT-/HFG-oracle rows use scalar oracle length caches for batch construction; reported throughput excludes cache construction.\quad
MMMU-MC scores are choice-likelihood scores over the 850/900 letter-labeled MMMU validation rows; the 50 non-letter rows are excluded by protocol.
\end{minipage}
\end{table}

\section{Auxiliary Instruction-Following Evaluation}
\label{app:ifeval}

As an additional supervised fine-tuning (SFT) quality check, we evaluate the same Full FT UltraChat and ShareGPT4o checkpoints on IFEval~\cite{zhou2023ifeval}, a rule-based instruction-following benchmark with verifiable text constraints.
This check is intended to complement the MMLU/MMMU-MC and validation-loss evidence in the main paper: MMLU-style multiple-choice accuracy is not the only proxy for SFT behavior, while IFEval measures whether a generated response follows explicit formatting and content constraints.
We use greedy decoding with max 1024 new tokens and the standard strict/loose IFEval rule checks.
The ShareGPT4o rows are checkpoints trained on multimodal data, but IFEval itself is text-only; we therefore interpret it as an instruction-following auxiliary evaluation rather than a visual/multimodal benchmark.
Table~\ref{tab:ifeval_appendix} reports this auxiliary check alongside the corresponding main-task score for the same checkpoint group.

\begin{table}[ht]
\centering
\caption{\textbf{Auxiliary IFEval instruction-following check.} Values are mean$\pm$std over three matched seeds, reported as percentages. Main is MMLU for UltraChat and MMMU-MC choice-likelihood accuracy for ShareGPT4o; IFEval uses 541 prompts and 834 instructions.}
\label{tab:ifeval_appendix}
\scriptsize
\setlength{\tabcolsep}{2.4pt}
\begin{tabular}{@{}lllrrrrr@{}}
\toprule
Scale & Data & Method & Main & P-strict & P-loose & I-strict & I-loose \\
\midrule
2B & UltraChat   & Standard & $58.17{\pm}0.06$ & $25.02{\pm}2.67$ & $29.82{\pm}3.15$ & $37.81{\pm}2.37$ & $42.97{\pm}2.43$ \\
2B & UltraChat   & ODB      & $58.98{\pm}0.18$ & $27.91{\pm}1.51$ & $30.75{\pm}1.30$ & $41.45{\pm}1.34$ & $44.16{\pm}1.20$ \\
2B & ShareGPT4o & Standard & $41.29{\pm}0.71$ & $45.78{\pm}1.41$ & $50.34{\pm}0.77$ & $56.83{\pm}1.18$ & $61.43{\pm}0.78$ \\
2B & ShareGPT4o & ODB      & $40.12{\pm}0.36$ & $52.19{\pm}1.02$ & $57.18{\pm}0.21$ & $63.39{\pm}0.39$ & $68.15{\pm}0.39$ \\
8B & UltraChat   & Standard & $72.85{\pm}0.40$ & $10.54{\pm}0.67$ & $13.43{\pm}0.11$ & $23.18{\pm}0.69$ & $27.10{\pm}0.42$ \\
8B & UltraChat   & ODB      & $74.75{\pm}0.11$ & $13.68{\pm}0.67$ & $16.51{\pm}1.05$ & $26.50{\pm}0.86$ & $30.82{\pm}1.05$ \\
8B & ShareGPT4o & Standard & $52.43{\pm}0.44$ & $71.53{\pm}0.67$ & $77.14{\pm}0.59$ & $79.62{\pm}0.48$ & $83.77{\pm}0.60$ \\
8B & ShareGPT4o & ODB      & $53.88{\pm}0.20$ & $77.82{\pm}0.67$ & $81.95{\pm}0.47$ & $84.33{\pm}0.30$ & $87.33{\pm}0.28$ \\
\bottomrule
\end{tabular}
\end{table}

On this auxiliary instruction-following check, ODB remains in the same quality band as Standard and has positive mean differences on the four IFEval metrics in each of the four matched settings.
At the seed level, prompt-level strict accuracy is positive for ODB in all 12 matched pairs, while prompt-level loose accuracy is positive in 11/12 pairs (the remaining 2B UltraChat seed has a small negative loose-score difference).
This auxiliary check does not indicate degraded SFT instruction-following behavior in these matched checkpoints, while the primary quality evidence remains the validation-loss and task-benchmark results reported in Tables~\ref{tab:public_fullft}, \ref{tab:multinode_fft}, and \ref{tab:public_lora}.

\section{Quality Hyperparameter Sensitivity}
\label{app:quality_hp}

The quality results in Tables~\ref{tab:public_fullft} and~\ref{tab:public_lora} use ODB configurations selected per dataset by the protocol in \S\ref{sec:eval}.
Since ODB fixes $\text{bs}=1$, the primary quality-sensitive knob is $L_{\max}$; \texttt{prefetch\_factor} controls outstanding depth and is treated as a throughput/overlap knob.
For low-CV datasets (e.g., LLaVA, CV=0.29), $L_{\max}$ significantly affects quality because most samples are short ($<$2048 tokens): a large $L_{\max}$ causes ODB to pack many short samples into a single step, creating an effective batch size that deviates from standard training.

Table~\ref{tab:hp_search} shows an auxiliary generated-answer MMMU sensitivity check for the LLaVA LoRA (8B) $L_{\max}$ sweep (one validation seed).
This auxiliary generated-answer setting is not mixed with the main MMMU-MC tables.
In this auxiliary generated-answer setting, $L_{\max}{=}4096$ matches standard training quality (28.2\% vs.\ 28.4\%), while the throughput-optimal configuration ($L_{\max}{=}16384$) is 1.3\,pp lower.

\begin{table}[ht]
\centering
\caption{Auxiliary LLaVA ODB $L_{\max}$ sensitivity check (Qwen3-VL-8B, LoRA, generated-answer MMMU validation, one seed). ODB bs=1 (fixed), \texttt{pf}=256 (default). Standard baseline: bs=8, generated-answer MMMU=28.4\%. Main multimodal benchmark tables use parser-free MMMU-MC.}
\label{tab:hp_search}
\footnotesize
\setlength{\tabcolsep}{4pt}
\begin{tabular}{@{}ccc@{}}
\toprule
$L_{\max}$ & Gen-answer MMMU & $\Delta$ vs.\ Std (pp) \\
\midrule
16384 & 27.1\% & $-$1.3\,pp \\
8192 & 27.6\% & $-$0.8\,pp \\
\textbf{4096} & \textbf{28.2\%} & \textbf{$-$0.2\,pp} \\
\bottomrule
\end{tabular}
\vspace{1mm}

{\scriptsize $L_{\max}$ controls the dynamic batch size via $B(l) = \lfloor L_{\max}/l \rfloor$; larger values pack more short samples per step, shifting the effective batch composition further from standard training. The 8192 result is the mean of two independent runs (27.7\%, 27.4\%). All main multimodal benchmark scores use MMMU-MC.}
\end{table}

The pattern is consistent: on low-CV data where samples are homogeneous, ODB should use a conservative $L_{\max}$ to keep the dynamic batch composition close to standard training.
Accordingly, the high-CV ShareGPT4o rows use the selected throughput configurations without a quality-specific $L_{\max}$ override, while remaining within the paper's Standard-comparable quality framing.

\section{Experimental Setup Details}
\label{app:setup_details}

Per-sample tokenized length statistics (full-pass under the Qwen3-VL tokenizer with vision tokens expanded at load time) and the resulting \texttt{cutoff\_len}:

\begin{table}[ht]
\centering
\caption{Per-sample tokenized length statistics. ShareGPT4o has the longest tail (Max 12K) and the highest CV.}
\label{tab:length_stats}
\footnotesize
\setlength{\tabcolsep}{4pt}
\begin{tabular}{@{}lccccccc@{}}
\toprule
Dataset & Samples & Mean & Median & P95 & P99 & Max & \texttt{cutoff\_len} \\
\midrule
UltraChat  & 207{,}865 & 1{,}196 & 1{,}104 & 2{,}239 & 3{,}065 & 4{,}471  & 8192 \\
LLaVA      & 157{,}712 &   508 &   463 &   814 &   914 & 1{,}260  & 2048 \\
ShareGPT4o &  57{,}284 & 1{,}511 &   977 & 4{,}584 & 8{,}937 & 12{,}110 & 16384 \\
\bottomrule
\end{tabular}
\end{table}

The 6 synthetic distributions used in correctness audits (1000 samples each): uniform-narrow $\mathcal{U}$[64,512], uniform-wide $\mathcal{U}$[64,2048], longtail (90\% short / 10\% long), bimodal (50/50), all-long $\mathcal{U}$[1800,2048], and all-short $\mathcal{U}$[32,64].

\paragraph{GMT/BMT distributed batching (\S\ref{sec:e2e}, \S\ref{sec:related}).}
GMT and BMT are rank-replicated oracle samplers: every rank computes the same global batch list (GMT: ascending-length sort plus greedy packing against a max-token budget; BMT: epoch-seeded shuffle, sample-count buckets, within-bucket length sort, greedy packing, then batch shuffle). Following fairseq-style max-token batching, feasibility is computed on padded token area, $\max_{i\in b} l_i \cdot |b| \le \texttt{max\_tokens\_budget}$, with singleton overflows allowed in our oracle implementation to preserve zero truncation and full-epoch coverage. The list is then padded to a multiple of the world size by wrap-around repetition of the leading batches and assigned to ranks by striding, guaranteeing identical step count on every rank. Last-batch handling is by construction: the wrap-around padding adds at most $W{-}1$ repeated batches per epoch, the offline analog of ODB's wrap-around padding (\S\ref{sec:alignment}). All ranks read the same scalar lengths array from the oracle cache during DataLoader construction, so no per-step length broadcast is required at training time.

\paragraph{Oracle length cache for GMT-oracle / BMT-oracle (\S\ref{sec:e2e}, \S\ref{sec:related}).}
To make the GMT / BMT comparison maximally favorable to the offline token-budget family, we precompute a per-(dataset, transform policy, template, and cutoff) scalar cache of $\texttt{len(input\_ids)}$ for every sample by a single forward pass through the LLaMA-Factory preprocessing/augmentation, templating, tokenization, and visual-token-expansion pipeline; no token IDs are stored or reused at training time. Construction-time precompute cost on a single H20 was 31\,s for UltraChat (207{,}865 samples; $\sim$6{,}700 sam/s, pure tokenization), 137\,s for ShareGPT4o (57{,}284 samples; $\sim$418 sam/s, image-blob $+$ tokenization), and 55\,min for LLaVA (157{,}712 samples; $\sim$48 sam/s, image-IO-bound under per-sample JPEG decode). The same runtime path therefore also appears during oracle cache construction, but the reported training throughput excludes this one-time cost; during training, GMT/BMT still perform normal online preprocessing, augmentation, tokenization, and visual-token expansion. The cache is invalidated and must be rebuilt on any preprocessing or augmentation-policy change; ODB requires no separate length precompute because it forms batches from lengths observed online in the DataLoader path.

\paragraph{MM-Mix composition (case study, Section~\ref{sec:case_study}).}
The production multimodal mixture aggregates 7 open-source datasets (272{,}589 unique samples; 2 training epochs $\Rightarrow$ 545{,}178 sample-views) spanning OCR, VQA, and image captioning. The bimodal length distribution (many short OCR/VQA labels alongside long captioning samples) yields CV${\approx}0.8$ and short-sample fraction $f_s{\approx}0.37$.

\begin{table}[ht]
\centering
\caption{MM-Mix composition. All datasets are publicly available under their original licenses; ``LO''~$=$~LLaVA-OneVision distribution.}
\label{tab:mmmix_composition}
\footnotesize
\setlength{\tabcolsep}{4pt}
\begin{tabular}{@{}llrl@{}}
\toprule
Dataset & Source & Samples & Task / Modality \\
\midrule
IIIT5K              & LO     &   1{,}990 & English scene-text OCR \\
ORAND-CAR-A         & LO     &   1{,}999 & Synthetic digit OCR \\
BCTR-Splice (scene) & BCTR   &  11{,}904 & Chinese scene-text OCR \\
A-OKVQA             & LLaVA  &  17{,}056 & Multiple-choice visual QA \\
VQAv2               & GeneralVQA &  82{,}783 & Open-ended visual QA \\
Image-Textualization (filtered) & LO &  99{,}573 & Image$\rightarrow$text captioning \\
ShareGPT4o (caption subset) & LO &  57{,}284 & Long-form captioning / dialogue \\
\midrule
\textbf{Total (unique)} & & \textbf{272{,}589} & 7 datasets, 2 epochs $=$ 545{,}178 sample-views \\
\bottomrule
\end{tabular}
\end{table}

\paragraph{8B MM-Mix throughput sweep (Section~\ref{sec:case_study}).}
On Qwen3-VL-8B-Instruct (8$\times$H20, 20-min profiling window per configuration), Standard bs$=$1 reaches 7.55\,sam/s, Sorted bs$=$2 reaches 15.80\,sam/s ($2.09\times$), and ODB peaks at $L_{\max}{=}8192$ with 22.07\,sam/s ($2.92\times$). The $L_{\max}$ sweep at default I/O (\texttt{pf}=256, \texttt{nw}=4) is single-peaked: $L_{\max}{=}4096$ gives $2.42\times$, $L_{\max}{=}6144$ gives $2.87\times$, $L_{\max}{=}8192$ gives $2.92\times$, $L_{\max}{=}12288$ gives $2.80\times$. Aggressive prefetch values (\texttt{pf}$\geq$1024, lifting outstanding-depth $D=\texttt{pf}\times\texttt{nw}$ to $\geq{}4096$) regressed sample throughput on this 8B workload---with $L_{\max}{=}4096$, \texttt{pf}=1024/2048/4096 decreasing to $1.58\times/1.29\times/0.98\times$---because larger in-flight buffers stress 8B activation memory through wider length-grouped batches that lengthen step time more than they raise per-step sample count, consistent with the first-order memory model discussed in Limitations. We therefore use the default ODB configuration as the starting point for this workload.

\paragraph{MM-Mix churn-inclusive accounting and benchmark quality.}
For churn-inclusive accounting, we distinguish \emph{churn-exclusive} training throughput, which excludes offline preparation, from a conservative \emph{cache-inclusive cost lower bound} that adds only the measured scalar oracle-cache construction time. The MM-Mix GMT/BMT oracle cache was built on one H20 for 272{,}589 unique samples in 299.9\,s (309\,s end-to-end wall-clock), producing the per-sample post-pipeline \texttt{len(input\_ids)} cache used for oracle batch construction. This charge is still favorable to offline methods: it does not separately account for broader production overheads such as sample-store scans outside the measured prepass, materializing and validating sorted/bucketed orders, staging metadata across workers, or rebuilding artifacts when the mixture, template, transform policy, or cutoff changes. ODB has no separate scalar-length precompute and forms batches from lengths observed online in the DataLoader path.

\begin{table}[ht]
\centering
\caption{\textbf{MM-Mix full-epoch 2B case-study results.} Two nodes $\times$ 8 H20 GPUs, Qwen3-VL-2B Full FT, 3-seed mean$\pm$std. \texttt{sam/s} is train-split emitted samples divided by wall-clock; for ODB we recompute literal throughput from the train split and runtime because the fixed-batch counter is not meaningful for dynamic batches. Score is MMMU-MC choice-likelihood accuracy.}
\label{tab:mmmix_quality}
\scriptsize
\setlength{\tabcolsep}{3pt}
\resizebox{\linewidth}{!}{%
\begin{tabular}{@{}lcccc@{}}
\toprule
Method & sam/s & Speedup & Score & Val Loss \\
\midrule
Standard & $17.85 \pm 0.15$ & $1.00\times$ & $43.33 \pm 2.24$ & $0.9674 \pm 0.0035$ \\
Sorted & $20.62 \pm 2.08$ & $1.15\times$ & $43.65 \pm 2.32$ & $1.4028 \pm 0.0312$ \\
GMT-oracle & $26.68 \pm 2.09$ & $1.49\times$ & $49.14 \pm 0.65$ & $0.9731 \pm 0.0028$ \\
BMT-oracle & $28.66 \pm 0.29$ & $1.61\times$ & $48.27 \pm 0.60$ & $0.9732 \pm 0.0032$ \\
HFG-oracle & $19.46 \pm 2.48$ & $1.09\times$ & $43.49 \pm 1.67$ & $0.9656 \pm 0.0028$ \\
\midrule
ODB & $79.15 \pm 4.16$ & $4.43\times$ & $46.31 \pm 0.44$ & $1.0137 \pm 0.0029$ \\
\bottomrule
\end{tabular}
}
\end{table}

\paragraph{Per-config ODB hyperparameters (Table~\ref{tab:public_fullft}, footnote~\textsuperscript{b}).}
The per-config $(L_{\max}, \texttt{pf}, \texttt{buffer})$ tuples used in Table~\ref{tab:public_fullft}:
\begin{itemize}[nosep,leftmargin=1.2em]
    \item UltraChat 8B: $(12288, 1024, 1024)$
    \item UltraChat 2B: $(16384, 1024, 1024)$
    \item LLaVA 8B: $(12288, 256, 1024)$
    \item LLaVA 2B: $(8192, 1024, 1024)$
    \item ShareGPT4o 8B: $(12288, 256, 1024)$
    \item ShareGPT4o 2B: $(4096, 256, 1024)$
\end{itemize}
These tuples are the selected full-epoch configurations after the speed-first profiling rule in \S\ref{sec:eval}; full-epoch quality is reported separately in Table~\ref{tab:public_fullft}. The resulting updates-per-epoch and batch-shape statistics are reported below rather than constrained to a fixed universal step-count target.

\paragraph{Throughput decomposition (8B Full FT).}
Table~\ref{tab:throughput_decomp_8b} decomposes throughput into the per-step factors that dynamic batching changes. The columns are reported per-cell as 3-seed means: \texttt{sam/s}~$=$~train-split emitted samples / wall-clock; \texttt{tok/s}~$=$~real unpadded tokens / wall-clock; \texttt{upd/ep}~$=$~optimizer updates per epoch; \texttt{sam/upd}~$=$~emitted samples/update and \texttt{tok/upd}~$=$~real unpadded tokens/update; \texttt{pad\%}~$=$~$1{-}\sum L_{\text{real}} / \sum L_{\text{compute}}$ (cumulative padding fraction); \texttt{dl-wait\%} and \texttt{compute\%}~$=$~unhidden DataLoader wait and GPU-compute fractions of elapsed wall time (remainder is pipeline overlap and other; we report the unhidden component because pipeline overlap, when high, makes raw \texttt{nvidia-smi} utilization a misleading proxy). HuggingFace Trainer's native \texttt{train\_samples\_per\_second} reports $\texttt{world\_batch}\times\texttt{updates}/\texttt{runtime}$, which double-counts under dynamic batching where each update consumes a variable real-sample count.

\begin{table}[ht]
\centering
\caption{\textbf{Throughput decomposition on 8B Full FT.} For the Standard, ODB, GMT-oracle, BMT-oracle, and HFG-oracle cells of the 8B side of Table~\ref{tab:public_fullft}, this table reports literal \texttt{sam/s}, updates per epoch, emitted samples/update, real unpadded tokens/update, cumulative padding, and unhidden DataLoader-wait / compute fractions. For dynamic-batch rows, \texttt{sam/s} is recomputed as train-split emitted samples divided by wall-clock time, matching Table~\ref{tab:public_fullft}.}
\label{tab:throughput_decomp_8b}
\footnotesize
\setlength{\tabcolsep}{4pt}
\begin{tabular}{@{}llrrrrrrrr@{}}
\toprule
Dataset & Method & sam/s & tok/s & upd/ep & sam/upd & tok/upd & pad\% & dl-wait\% & compute\% \\
\midrule
UltraChat & Standard   &  5.77 &  6{,}901 & 24{,}684 &   8.00 &   9{,}566 &  0.0 & 0.01 & 99.9 \\
          & ODB        & 10.23 & 12{,}234 &  2{,}692 &  73.36 &  87{,}725 &  1.3 & 0.00 & 99.8 \\
          & GMT-oracle & 10.94 & 13{,}087 &  1{,}890 & 104.48 & 124{,}981 &  0.0 & 0.00 & 99.8 \\
          & BMT-oracle & 10.31 & 12{,}346 &  1{,}901 & 103.90 & 124{,}205 &  0.5 & 0.00 & 99.8 \\
          & HFG-oracle &  7.33 &  9{,}127 & 12{,}342 &  16.00 &  19{,}932 & 12.8 & 0.01 & 99.9 \\
\midrule
LLaVA     & Standard   & 14.38 &  7{,}298 &  2{,}342 &  63.97 &  32{,}469 & 32.5 & 0.00 & 99.7 \\
          & ODB        & 24.87 & 12{,}621 &    844 & 177.59 &  90{,}133 &  2.1 & 1.62 & 97.5 \\
          & GMT-oracle & 26.65 & 13{,}549 &    591 & 253.51 & 128{,}886 &  0.0 & 0.00 & 99.3 \\
          & BMT-oracle & 25.70 & 13{,}068 &    597 & 250.97 & 127{,}597 &  0.6 & 0.00 & 99.3 \\
          & HFG-oracle & 21.84 & 11{,}184 &  1{,}171 & 127.95 &  65{,}523 & 15.8 & 0.00 & 99.5 \\
\midrule
ShareGPT4o & Standard   &  2.37 &  3{,}535 &  6{,}803 &   8.00 &  11{,}949 &  0.0 & 0.01 & 99.9 \\
           & ODB        &  5.83 &  8{,}705 &  1{,}030 &  52.82 &  78{,}891 &  0.9 & 1.78 & 97.3 \\
           & GMT-oracle &  2.57 &  3{,}876 &  4{,}762 &  11.43 &  17{,}236 &  1.8 & 0.00 & 99.8 \\
           & BMT-oracle &  2.50 &  3{,}795 &  6{,}057 &   8.98 &  13{,}619 &  0.0 & 0.01 & 99.8 \\
           & HFG-oracle &  2.82 &  4{,}323 &  6{,}803 &   8.00 &  12{,}284 &  0.0 & 0.01 & 99.8 \\
\bottomrule
\end{tabular}
\end{table}

Three patterns are visible. (i)~\emph{Padding is the binding constraint on LLaVA, sample-count on UltraChat / ShareGPT4o.} Standard hits $32.5\%$ padding on LLaVA, whereas bs=1 Standard avoids padding on UltraChat and ShareGPT4o at the cost of only eight samples/update. ODB attacks the limiting factor in each regime: it reduces LLaVA padding to $2.1\%$ and raises \texttt{sam/upd} to 177.6, while increasing \texttt{sam/upd} by $9.2\times$ on UltraChat and $6.6\times$ on ShareGPT4o. (ii)~\emph{Oracle baselines differ in update geometry.} On LLaVA, GMT/BMT-oracle process roughly 251--254 samples/update and only 591--597 updates per epoch, versus ODB's 177.6 samples/update and 844 updates; HFG's randomized fixed-batch construction instead uses 128.0 samples/update with 15.8\% padding. On ShareGPT4o, HFG falls back to the bs=1 shape, while ODB reaches 52.8 samples/update. These regimes are interpreted jointly with validation loss and MMMU-MC rather than as throughput alone. (iii)~\emph{Pipeline starvation is not the dominant explanation.} Standard and oracle rows are compute-bound, and ODB's multimodal rows show only a small unhidden DataLoader-wait fraction ($\leq1.78\%$); thus the throughput differences primarily reflect \emph{batch shape} and update geometry, not raw I/O efficiency.

\paragraph{Throughput decomposition (2B Full FT).}
Table~\ref{tab:throughput_decomp_2b} reports the same decomposition for the 2B side of Table~\ref{tab:public_fullft}. It uses the same literal-throughput convention: \texttt{sam/s} is train-split emitted samples divided by wall-clock time, while token/update and DataLoader-wait fields are aggregated from the corresponding full-epoch runs.

\begin{table}[ht]
\centering
\caption{\textbf{Throughput decomposition on 2B Full FT.} For the Standard, ODB, GMT-oracle, BMT-oracle, and HFG-oracle cells of the 2B side of Table~\ref{tab:public_fullft}, this table reports literal \texttt{sam/s}, updates per epoch, emitted samples/update, real unpadded tokens/update, cumulative padding, and unhidden DataLoader-wait / compute fractions. Rows are 3-seed means over full-epoch runs.}
\label{tab:throughput_decomp_2b}
\footnotesize
\setlength{\tabcolsep}{4pt}
\begin{tabular}{@{}llrrrrrrrr@{}}
\toprule
Dataset & Method & sam/s & tok/s & upd/ep & sam/upd & tok/upd & pad\% & dl-wait\% & compute\% \\
\midrule
UltraChat & Standard   & 20.98 & 25{,}087 & 24{,}684 &   8.00 &   9{,}566 & 0.0 & 0.04 & 99.9 \\
          & ODB        & 36.91 & 44{,}129 &  2{,}003 &  98.59 & 117{,}886 & 1.7 & 0.01 & 99.6 \\
          & GMT-oracle & 39.44 & 47{,}153 &  3{,}983 &  49.58 &  59{,}272 & 0.0 & 0.01 & 99.8 \\
          & BMT-oracle & 35.84 & 42{,}848 &  1{,}901 & 103.90 & 124{,}224 & 0.4 & 0.00 & 99.8 \\
          & HFG-oracle & 27.28 & 33{,}447 & 24{,}684 &   8.00 &   9{,}808 & 0.0 & 0.05 & 99.8 \\
\midrule
LLaVA     & Standard   & 47.92 & 24{,}319 &  4{,}683 &  31.99 &  16{,}238 & 24.2 & 0.02 & 99.6 \\
          & ODB        & 82.42 & 41{,}830 &  1{,}259 & 119.01 &  60{,}399 &  1.6 & 0.01 & 98.5 \\
          & GMT-oracle & 79.44 & 40{,}375 &  2{,}496 &  60.02 &  30{,}505 &  0.0 & 0.02 & 99.3 \\
          & BMT-oracle & 75.64 & 38{,}456 &  2{,}502 &  59.89 &  30{,}448 &  0.1 & 0.01 & 99.4 \\
          & HFG-oracle & 69.52 & 35{,}837 &  2{,}342 &  63.97 &  32{,}977 & 14.5 & 0.01 & 99.4 \\
\midrule
ShareGPT4o & Standard   &  6.51 &  9{,}717 &  6{,}803 &   8.00 &  11{,}949 & 0.0 & 0.01 & 99.9 \\
           & ODB        & 16.09 & 24{,}027 &  2{,}719 &  20.01 &  29{,}892 & 0.4 & 3.08 & 95.5 \\
           & GMT-oracle &  7.03 & 10{,}599 &  4{,}762 &  11.43 &  17{,}236 & 1.8 & 0.01 & 99.8 \\
           & BMT-oracle &  7.03 & 10{,}731 &  4{,}763 &  11.43 &  17{,}447 & 2.1 & 0.01 & 99.9 \\
           & HFG-oracle &  7.71 & 11{,}833 &  6{,}803 &   8.00 &  12{,}284 & 0.0 & 0.01 & 99.8 \\
\bottomrule
\end{tabular}
\end{table}

The 2B decomposition shows the same batch-shape mechanism as the 8B table, with scale-specific details. Standard's bottleneck again flips between sam/upd (UltraChat / ShareGPT4o, bs=1) and padding (LLaVA, 24.2\%). ODB attacks the cell-specific bottleneck: on LLaVA it reduces padding to 1.6\% and raises samples/update from 32.0 to 119.0; on UltraChat and ShareGPT4o it raises the fixed-bs Standard rows from eight samples/update to 98.6 and 20.0 respectively. The oracle rows achieve competitive throughput through offline length-aware grouping, but with different update geometry: on LLaVA, GMT/BMT use about 60 samples/update, HFG keeps the fixed-batch shape at 64.0 samples/update with 14.5\% padding, and ODB uses 119.0 samples/update; on UltraChat, HFG keeps the bs=1 shape while GMT/BMT widen updates to 49.6/103.9 samples. DataLoader wait remains negligible except for ShareGPT4o ODB (3.1\%), where throughput is still dominated by the larger emitted-sample/update count rather than by I/O starvation. Thus the 2B table sharpens the 8B conclusion: ODB's gain is a batch-shape effect, not a Trainer-accounting artifact or raw I/O speed.

\section{HFG-oracle Randomized Fixed-Batch Baseline}
\label{app:hfg_baseline}

HFG-oracle instantiates the HuggingFace \texttt{group\_by\_length} family as a randomized fixed-batch oracle baseline. Each epoch samples a random permutation, partitions it into megabatches, sorts each megabatch by the oracle post-tokenization length, concatenates the megabatches, pads the index list to a multiple of world size, and stride-shards it across ranks. It uses the same scalar length cache as GMT/BMT-oracle but keeps a fixed batch size, separating epoch-level randomization from max-token scheduling while avoiding globally sorted epoch order. For HFG-oracle, bs is selected as the largest full-epoch-safe fixed batch size under the same profiling and full-epoch survival protocol used for Sorted; speedups normalize to the matching Standard rows in Table~\ref{tab:public_fullft}.

The HFG-oracle rows are reported directly in Table~\ref{tab:public_fullft} rather than duplicated in a second appendix table. They show that randomized length grouping controls for fully sorted epoch order without changing the main conclusions: under MMMU-MC, HFG remains in the same quality band as the other non-sorted LLaVA methods, but its fixed batch size cannot exploit the high-CV ShareGPT4o tail where ODB is $2.07\times$ faster at 8B (5.83 vs. 2.82 sam/s) and $2.09\times$ faster at 2B (16.09 vs. 7.71 sam/s). Like GMT/BMT-oracle, HFG inherits the cache-rebuild limitation under preprocessing, template, cutoff, or augmentation changes, whereas ODB observes lengths online after those transformations.

\section{CV/\texorpdfstring{$f_s$}{fs} Two-Feature Decomposition: Phenomenological Reference}
\label{app:cv_fs_decomp}

Section~\ref{sec:guidance} uses CV and the short-sample fraction $f_s$ qualitatively. We report here the explicit two-anchor pinning that motivated their selection, together with the methodological caveat that prevents us from positioning it as a predictive model.

The minimal two-feature linear form,
\begin{equation}
\hat{S}(\text{CV}, f_s) \;=\; 1 \;+\; \alpha\,\text{CV} \;+\; \beta\,f_s,
\label{eq:speedup_estimator}
\end{equation}
pinned on the two 2B Full FT workloads with both features measured (ShareGPT4o: $\text{CV}{=}1.00$, $f_s{\approx}0.01$, $S{=}2.47$; MM-Mix: $\text{CV}{=}0.80$, $f_s{\approx}0.37$, $S{=}4.43$), gives $\alpha{\approx}1.41$, $\beta{\approx}6.23$.

\paragraph{Scope of the two-feature fit.} $(\text{CV}, f_s)$ is a workload/configuration descriptor: $\text{CV}$ summarizes the tokenized length distribution, while $f_s=\Pr[\ell<L_{\max}/4]$ measures short-sample mass under the selected token budget. The headline cells are dominated by four dataset-level workload families rather than 14 independent locations in the $(\text{CV}, f_s)$ plane: 8B Full FT $\times 3$ + 2B Full FT $\times 3$ + 8B LoRA $\times 3$ + 2B LoRA $\times 3$ + MM-Mix at 2B/8B vary model scale or finetuning regime around a small set of length-distribution families. Treating all cells as independent would mostly add replication noise from model scale and finetuning regime. A leave-one-out cross-validation (LOOCV) $R^2$ computed over those cells would reflect cell-noise variance, not the predictive power of the two-feature form.

\paragraph{Scope.} Eq.~\ref{eq:speedup_estimator} should be read as a phenomenological reference within the calibrated range $\text{CV}\in[0.80, 1.00]$, $f_s\in[0.01, 0.37]$; it captures the separation between MM-Mix and ShareGPT4o that CV alone cannot, but it is neither a predictor nor a standalone contribution. Section~\ref{sec:guidance} instead uses CV ranking plus the $f_s$ deviation screen as qualitative guidance for the deployment recipes (ROI screen, outstanding-depth tuning, $L_{\max}$ binding).

\section{Auxiliary Generated-Answer MMMU Format-Degradation Analysis}
\label{app:sorted_format}

We empirically test the answer-format-degradation hypothesis (Section~\ref{sec:quality}) for LLaVA 8B Full FT Sorted, whose generated-answer MMMU score collapses to $5.52{\pm}0.18\%$ despite val\_loss being only $+11.6\%$ above Standard.
This appendix explains a parser-sensitive analysis and motivates the parser-free MMMU-MC protocol used in the main tables; it is not a main benchmark result.
We sample 120 MMMU validation questions across four subjects (Art, Math, Computer Science, History) and generate raw model outputs from representative LLaVA 8B Full FT checkpoints: Standard (bs$=$8), Sorted (bs$=$16; the longest-tail-safe fixed batch after bs$=$32 OOMs), and ODB at a conservative auxiliary setting ($L_{\max}{=}4096$, \texttt{pf}=1024). These generated-answer checkpoints are used only to study parser sensitivity; the main benchmark cells use MMMU-MC likelihood scoring. Decoding is identical greedy with max 256 new tokens; Table~\ref{tab:sorted_format_diag} reports response-length and answer-format statistics.

\begin{table}[ht]
\centering
\caption{Raw-output analysis on 120 MMMU samples under the generated-answer protocol. \emph{1-letter\%} is the fraction of stripped responses that are exactly a single A--H letter; \emph{verbose\%} is responses $>$ 20 chars after \texttt{<think>} stripping; \emph{extracted-acc\%} is the subset accuracy under the diagnostic answer-extraction rule.}
\label{tab:sorted_format_diag}
\footnotesize
\setlength{\tabcolsep}{6pt}
\begin{tabular}{@{}lrrrr@{}}
\toprule
Method & mean chars & 1-letter\% & verbose\% & extracted-acc\% \\
\midrule
Standard (bs=8)               & 102.5  & 75.0\% & 23.3\% & 22.5\% \\
\textbf{Sorted (bs=16)}       & 110.3  & \textbf{9.2\%}  & \textbf{90.8\%} & \textbf{3.3\%} \\
ODB                            & 115.6  & 76.7\% & 21.7\% & 20.0\% \\
\bottomrule
\end{tabular}
\end{table}

The 120-sample subset accuracies (Std $22.5\%$, ODB $20.0\%$, Sorted $3.3\%$) track the generated-answer full-evaluation pattern from the same checkpoints (Std $22.30\%$, ODB $22.00\%$, Sorted $5.52\%$) in rank ordering and gap magnitude. These generated-answer numbers are not used in the main Score columns, which use MMMU-MC likelihood.

\paragraph{Format-degradation pattern.} Standard and ODB answer in MMMU's expected single-letter format $\sim$$76\%$ of the time; \textbf{Sorted answers in single-letter format only $9.2\%$ of the time}, with $90.8\%$ of responses being verbose continuations. Inspecting the verbose Sorted outputs, the failure mode is not ``the model answers correctly with extra explanation''---instead, the model degenerates into echoing chat-template structure or asking a follow-up question rather than producing the multiple-choice answer.

\paragraph{Representative raw outputs (Art subject, ground truth ``C'').}
\begin{itemize}\setlength{\itemsep}{0pt}
\item \textbf{Standard / ODB}: \texttt{"C"} (single letter, MMMU regex extracts ``C'' $\to$ correct).
\item \textbf{Sorted}: \texttt{"user$\backslash$nWhat is the subject of the painting?"} (template echo + follow-up question; no valid A--H answer is extracted, yielding an incorrect prediction).
\end{itemize}
This pattern is consistent across all four subjects (Art, Math, Computer Science, History): verbose Sorted responses on MMMU questions are typically meta-questions or partial assistant utterances rather than usable answer letters.

\paragraph{Conclusion.} These measurements support the format-degradation hypothesis: Sorted's val\_loss is only $+11.6\%$ above Standard because next-token likelihood on natural-language continuations remains plausible, but \textbf{$>90\%$ of Sorted's generated MMMU responses are not in the answer format MMMU's exact-match scoring expects}. The generated-answer MMMU drop from $22.30\%$ (Std) to $5.52\%$ (Sorted) therefore primarily reflects output-format drift under parser-extracted scoring, not catastrophic divergence in language modeling. ODB's length-\emph{grouped} (rather than length-\emph{sorted}) formulation does not exhibit this drift ($76.7\%$ single-letter rate, on par with Standard $75.0\%$) in this auxiliary setting. The main evaluated throughput--quality comparison in Section~\ref{sec:quality} therefore uses parser-free MMMU-MC likelihood rather than parser-extracted generated answers; this appendix explains why the generated-answer analysis is not the benchmark result.

\section{LoRA Results}
\label{app:lora_results}

\begin{table}[H]
\centering
\caption{\textbf{LoRA Results.} Same comparison under LoRA fine-tuning (rank=8, target=all). 8B and 2B LoRA: 3-seed mean$\pm$std on the same 8$\times$H20 setup as Table~\ref{tab:public_fullft}. Bold marks the highest emitted-sample throughput among online/no-cache rows; offline-oracle throughput cells are unbolded comparators. Score and Val Loss are reported as reference quality checks without bolding.}
\label{tab:public_lora}
\scriptsize
\setlength{\tabcolsep}{1.2pt}
\resizebox{\linewidth}{!}{%
\begin{tabular}{@{}llcccc@{\hspace{4pt}}cccc@{}}
\toprule
 & & \multicolumn{4}{c}{\textbf{8B (LoRA)}} & \multicolumn{4}{c}{\textbf{2B (LoRA)}} \\
\cmidrule(lr){3-6} \cmidrule(lr){7-10}
Dataset & Method & sam/s & Speedup & Score & Val Loss\textsuperscript{a} & sam/s & Speedup & Score & Val Loss \\
\midrule
\multirow{7}{*}{\shortstack[l]{UltraChat\\CV=0.48\\MMLU}}
& Standard     & $7.87 \pm 0.00$            & 1.00$\times$         & $76.12 \pm 0.06$\%        & $0.858 \pm 0.001$ & $25.16 \pm 0.07$ & 1.00$\times$ & $59.90 \pm 0.05$\% & $1.021 \pm 0.003$\textsuperscript{b} \\
& Sorted       & $10.15 \pm 0.01$ & 1.29$\times$ & $76.25 \pm 0.01$\%        & $0.877 \pm 0.001$ & $32.27 \pm 0.04$ & 1.28$\times$ & $60.04 \pm 0.03$\% & $1.043 \pm 0.003$ \\
& Packing      & $12.57 \pm 0.01$           & 1.60$\times$         & $76.49 \pm 0.04$\% & $1.122 \pm 0.000$ & $42.14 \pm 0.06$ & 1.68$\times$ & $60.61 \pm 0.05$\% & $1.068 \pm 0.003$ \\
& GMT-oracle\textsuperscript{c} & $13.48 \pm 0.02$ & 1.71$\times$ & $76.27 \pm 0.04$\% & $0.867 \pm 0.001$ & $45.49 \pm 0.06$ & 1.81$\times$ & $59.99 \pm 0.05$\% & $1.034 \pm 0.003$ \\
& BMT-oracle\textsuperscript{c} & $12.27 \pm 0.02$ & 1.56$\times$ & $76.36 \pm 0.03$\% & $0.875 \pm 0.001$ & $40.25 \pm 0.10$ & 1.60$\times$ & $60.25 \pm 0.04$\% & $1.061 \pm 0.003$ \\
& HFG-oracle\textsuperscript{c} & $10.79 \pm 0.03$ & 1.37$\times$ & $76.12 \pm 0.13$\% & $0.858 \pm 0.001$ & $29.81 \pm 0.09$ & 1.18$\times$ & $60.05 \pm 0.08$\% & $1.030 \pm 0.003$ \\
& \textbf{ODB} & $\mathbf{12.45 \pm 0.03}$ & \textbf{1.58$\times$} & $76.34 \pm 0.07$\% & $0.888 \pm 0.001$ & $\mathbf{41.63 \pm 0.11}$ & \textbf{1.65$\times$} & $60.10 \pm 0.10$\% & $1.045 \pm 0.003$ \\
\midrule
\multirow{6}{*}{\shortstack[l]{LLaVA\\CV=0.29\\MMMU-MC}}
& Standard     & $18.95 \pm 0.04$                          & 1.00$\times$         & $55.76 \pm 0.48$\% & $1.098 \pm 0.001$ & $53.98 \pm 0.05$ & 1.00$\times$        & $45.02 \pm 0.65$\%  & $1.299 \pm 0.003$ \\
& Sorted       & $24.91 \pm 0.03$\textsuperscript{d}       & 1.31$\times$         & $56.39 \pm 0.41$\%          & $1.221 \pm 0.002$ & $74.04 \pm 0.30$ & 1.37$\times$ & $44.82 \pm 0.12$\%\textsuperscript{e} & $1.339 \pm 0.003$\textsuperscript{e} \\
& GMT-oracle\textsuperscript{c} & $32.34 \pm 0.08$ & 1.71$\times$ & $56.82 \pm 0.43$\% & $1.135 \pm 0.001$ & $94.77 \pm 0.53$ & 1.76$\times$ & $43.84 \pm 0.38$\% & $1.376 \pm 0.004$ \\
& BMT-oracle\textsuperscript{c} & $29.99 \pm 0.02$ & 1.58$\times$ & $56.63 \pm 0.25$\% & $1.121 \pm 0.001$ & $90.86 \pm 0.34$ & 1.68$\times$ & $44.04 \pm 0.38$\% & $1.374 \pm 0.004$ \\
& HFG-oracle\textsuperscript{c} & $26.67 \pm 0.05$ & 1.41$\times$ & $57.17 \pm 0.31$\% & $1.124 \pm 0.001$ & $79.86 \pm 0.08$ & 1.48$\times$ & $43.84 \pm 0.07$\% & $1.331 \pm 0.003$ \\
& \textbf{ODB} & $\mathbf{30.98 \pm 0.07}$                & \textbf{1.63$\times$} & $56.47 \pm 0.00$\%          & $1.151 \pm 0.001$ & $\mathbf{94.60 \pm 1.01}$ & \textbf{1.75$\times$}        & $44.63 \pm 0.18$\%           & $1.301 \pm 0.003$ \\
\midrule
\multirow{6}{*}{\shortstack[l]{ShareGPT4o\\CV=1.00\\MMMU-MC}}
& Standard     & $2.83 \pm 0.01$            & 1.00$\times$         & $54.24 \pm 0.00$\%       & $1.217 \pm 0.006$ & $7.08 \pm 0.02$ & 1.00$\times$ & $41.77 \pm 0.20$\% & $1.346 \pm 0.005$\textsuperscript{b} \\
& Sorted       & $2.92 \pm 0.00$           & 1.03$\times$         & $54.55 \pm 0.41$\%       & $1.235 \pm 0.005$ & $7.30 \pm 0.03$ & 1.03$\times$ & $42.47 \pm 0.12$\% & $1.357 \pm 0.005$ \\
& GMT-oracle\textsuperscript{c} & $2.98 \pm 0.02$ & 1.05$\times$ & $54.55 \pm 0.27$\% & $1.218 \pm 0.006$ & $7.44 \pm 0.03$ & 1.05$\times$ & $41.92 \pm 0.24$\% & $1.347 \pm 0.005$ \\
& BMT-oracle\textsuperscript{c} & $2.98 \pm 0.02$ & 1.05$\times$ & $54.35 \pm 0.20$\% & $1.218 \pm 0.006$ & $7.46 \pm 0.03$ & 1.05$\times$ & $41.76 \pm 0.20$\% & $1.347 \pm 0.005$ \\
& HFG-oracle\textsuperscript{c} & $3.40 \pm 0.01$ & 1.20$\times$ & $54.19 \pm 0.27$\% & $1.217 \pm 0.006$ & $8.28 \pm 0.03$ & 1.17$\times$ & $41.53 \pm 0.31$\% & $1.346 \pm 0.005$ \\
& \textbf{ODB} & $\mathbf{6.81 \pm 0.04}$ & \textbf{2.41$\times$} & $54.98 \pm 0.18$\% & $1.235 \pm 0.006$ & $\mathbf{17.75 \pm 0.23}$ & \textbf{2.51$\times$} & $44.16 \pm 0.44$\% & $1.382 \pm 0.004$ \\
\bottomrule
\end{tabular}}

\vspace{1mm}
\begin{minipage}{\linewidth}
\raggedright\scriptsize
\textsuperscript{a}8B LoRA throughput, benchmark, and Val Loss use 3-seed LoRA means. Packing uses the packed-sequence denominator as in Table~\ref{tab:public_fullft}, footnote~\textsuperscript{d}.\\
\textsuperscript{b}2B LoRA UltraChat/ShareGPT4o throughput, downstream benchmark, and Val Loss are 3-seed means; Val Loss uses the same checkpoints as the throughput/benchmark aggregation.\\
\textsuperscript{c}GMT-/BMT-/HFG-oracle rows use the scalar oracle length cache for batch construction; speedups normalize to the matching Standard row in this table.\\
\textsuperscript{d}LLaVA 8B LoRA Sorted: full 3-seed mean; I/O sensitivity is discussed in Section~\ref{sec:io_sensitivity}.\\
MMMU-MC uses choice-likelihood scoring on the 850 letter-labeled validation rows. ODB remains in the same Standard-comparable band under LoRA on LLaVA (8B: $56.47\%$ vs $55.76\%$; 2B: $44.63\%$ vs $45.02\%$), is nominally above Standard on ShareGPT4o at both scales, and the 8B UltraChat ODB row reports MMLU $76.34\pm0.07\%$.\\
\textsuperscript{e}2B LoRA LLaVA throughput, Score, and Val Loss are aggregated from the 3-seed LoRA evaluation. The main multimodal comparison uses parser-free MMMU-MC and should be read together with Val Loss.
\end{minipage}
\end{table}

\section{Additional Ablations: Buffer Size and Loss Scaling Mode}
\label{app:more_ablations}

\paragraph{Buffer size.}
The grouping buffer determines how many samples the collate worker accumulates before forming groups; larger buffers enable tighter length-based grouping.
Table~\ref{tab:ablation_buf} sweeps the buffer on ShareGPT4o (CV=1.00, the most grouping-pressured dataset). Throughput improves sharply up to \texttt{buffer}=500 and remains in the high-throughput range around 1024--2000 on 2B (16.77--17.10\,sam/s) and peaks at 1024 on 8B (6.38\,sam/s), while padding is at or below 0.6\% at \texttt{buffer}$\geq$1024. Lower-CV datasets place weaker demands on the buffer, so we do not ablate them separately. This profiling sweep is used for throughput/shape diagnosis; full quality claims remain in the main and LoRA result tables, and the default \texttt{buffer}=1024 remains a low-padding main-sweep setting.

\begin{table}[ht]
\centering
\caption{Ablation: buffer size on ShareGPT4o profiling windows ($L_{\max}{=}4096$ for 2B, $8192$ for 8B; 8$\times$H20, single seed). Speedups normalize to the matching 20-minute Standard profiling baseline.}
\label{tab:ablation_buf}
\footnotesize
\setlength{\tabcolsep}{5pt}
\begin{tabular}{@{}llcccc@{}}
\toprule
Scale & Buffer & padding\% & sam/s & vs Std & Loss \\
\midrule
2B & 10 & 3.0\% & 8.46 & 1.28$\times$ & 1.364 \\
2B & 50 & 1.6\% & 12.20 & 1.84$\times$ & 1.375 \\
2B & 100 & 1.0\% & 13.73 & 2.07$\times$ & 1.350 \\
2B & 500 & 0.5\% & 15.69 & 2.37$\times$ & 1.346 \\
2B & 1024$^{*}$ & 0.4\% & 16.77 & 2.53$\times$ & 1.341 \\
2B & 2000 & 0.3\% & \textbf{17.10} & 2.58$\times$ & 1.366 \\
8B & 10 & 7.8\% & 2.87 & 1.25$\times$ & 1.285 \\
8B & 50 & 4.5\% & 4.21 & 1.83$\times$ & 1.283 \\
8B & 100 & 2.8\% & 5.03 & 2.19$\times$ & 1.283 \\
8B & 500 & 0.9\% & 5.25 & 2.29$\times$ & 1.274 \\
8B & 1024$^{*}$ & 0.6\% & \textbf{6.38} & 2.78$\times$ & 1.269 \\
8B & 2000 & 0.5\% & 5.98 & 2.61$\times$ & 1.268 \\
\bottomrule
\end{tabular}
\vspace{1mm}

{\scriptsize $^{*}$default configuration.}
\end{table}

\paragraph{Loss scaling mode.}
ODB supports three gradient scaling strategies:
(1)~\emph{Sample-level} $\mathcal{L}_{\text{scaled}} = \mathcal{L} \cdot (n_{\text{local}}/n_{\text{total}}) \cdot W$, with $n_{\text{total}}$ piggybacked on the first \texttt{all\_gather};
(2)~\emph{Approximate token-level} same form with $t_{\text{local}}/t_{\text{total}}$, post-alignment tokens estimated as $t_{\text{adj}} \approx n_{\text{adj}} \bar{t}$;
(3)~\emph{Exact token-level} uses the primary counts when alignment is a no-op and otherwise performs a deterministic second \texttt{all\_gather} of one \texttt{max\_groups}-length token-count vector per rank to re-broadcast post-alignment counts.
All main experiments use mode~(3); Table~\ref{tab:loss_scaling} shows that the three modes have similar throughput in short profiling windows, with differences within about 0.2\% on 2B and 1.0\% on 8B relative to sample-level scaling. Exact token-level scaling remains the conservative default for final runs because it satisfies Eq.~\ref{eq:exactness} bit-precisely.

\begin{table}[ht]
\centering
\caption{Loss scaling ablation (ShareGPT4o profiling windows, 8$\times$H20, single seed).}
\label{tab:loss_scaling}
\footnotesize
\setlength{\tabcolsep}{5pt}
\begin{tabular}{@{}llcccc@{}}
\toprule
Scale & Scaling mode & Loss & Unscaled loss & sam/s & Speed vs sample \\
\midrule
2B & Sample-level & 1.342 & 1.317 & 16.70 & --- \\
2B & Approx token & 1.292 & 1.317 & 16.70 & +0.0\% \\
2B & Exact token & 1.341 & 1.317 & 16.73 & +0.2\% \\
8B & Sample-level & 1.269 & 1.283 & 6.39 & --- \\
8B & Approx token & 1.211 & 1.279 & 6.45 & +1.0\% \\
8B & Exact token & 1.269 & 1.283 & 6.37 & -0.3\% \\
\bottomrule
\end{tabular}
\end{table}

\section{DataLoader I/O Sensitivity (Full Sweep)}
\label{app:io_sensitivity}

\begin{table}[H]
\centering
\caption{I/O sensitivity: single-seed 20-minute profiling throughput (sam/s) vs.\ \texttt{num\_workers} for Standard and default-join ODB across datasets and scales (8$\times$H20). This diagnostic sweep isolates input-pipeline depth at fixed configurations; main-table cells use the full-epoch protocol in \S\ref{sec:eval}.}
\label{tab:io}
\footnotesize
\setlength{\tabcolsep}{2pt}
\begin{tabular}{@{}l|ccccc|ccccc@{}}
\toprule
& \multicolumn{5}{c|}{Standard} & \multicolumn{5}{c}{ODB} \\
\texttt{nw} & 0 & 1 & 2 & 4 & 8 & 0 & 1 & 2 & 4 & 8 \\
\midrule
2B UltraChat & 20.6 & 20.9 & 20.9 & 20.9 & 20.9 & 36.8 & 36.5 & 36.7 & 36.8 & 37.0 \\
2B LLaVA & 44.7 & 48.7 & 48.9 & 48.6 & 48.4 & 80.2 & 72.7 & 78.3 & 80.6 & 85.1 \\
2B ShareGPT4o & 4.9 & 6.6 & 6.6 & 6.6 & 6.5 & 16.8 & 16.1 & 16.3 & 16.7 & 17.0 \\
8B UltraChat & 5.7 & 5.8 & 5.8 & 5.8 & 5.8 & 9.9 & 10.2 & 9.9 & 9.9 & 9.9 \\
8B LLaVA & 14.0 & 14.4 & 14.4 & 14.4 & 14.0 & 25.1 & 24.2 & 24.8 & 25.2 & 25.6 \\
8B ShareGPT4o & 2.1 & 2.3 & 2.3 & 2.3 & 2.3 & 6.4 & 6.4 & 6.4 & 6.4 & 6.4 \\
\bottomrule
\end{tabular}
\end{table}

Across the sweep, Standard is mostly flat after one or two workers, and ODB remains faster than Standard at every worker count. The strongest worker sensitivity is the expected \texttt{nw}=0 penalty for fixed-batch Standard on ShareGPT4o and a small ODB dip at \texttt{nw}=1--2 on LLaVA; by \texttt{nw}=4 the ODB rows are within about 5\% of their best values. This supports the operational rule: keep \texttt{nw}$\geq$4 as a portable default, start from \texttt{prefetch\_factor}=256, and tune per configuration rather than assuming that more workers alone explain ODB's throughput.

\section{Outstanding Depth Clamp Validation}
\label{app:clamp}

Section~\ref{sec:ablation} defines the outstanding depth as $D = \max(\texttt{pf} \times \texttt{nw}, \texttt{buffer\_size})$. When $\texttt{pf} \times \texttt{nw} < \texttt{buffer\_size}$, ODB's reset logic injects extra indices into the worker queue so that the collate process can assemble a full group---effectively clamping the in-flight sample count to $\texttt{buffer\_size}$. This appendix empirically confirms the clamp behaviour.

Fixing \texttt{buffer\_size}=1024 and \texttt{nw}=4, we measured \texttt{pipeline\_overlap} for $\texttt{pf} \in \{32, 64, 128\}$ on all three datasets and both scales. Every such \texttt{pf} has nominal depth $\texttt{pf} \times \texttt{nw} \in \{128, 256, 512\}$, all strictly below the buffer, so all three points share the same effective $D{=}1024$. Table~\ref{tab:clamp} reports the observed variation under this clamp.

\begin{table}[ht]
\centering
\caption{Clamp validation: \texttt{pipeline\_overlap} for \texttt{pf}$<256$ (\texttt{nw}=4, \texttt{buffer}=1024). The three clamped \texttt{pf} values share the same effective depth; small differences reflect profiling noise/workload variation rather than a change in effective depth.}
\label{tab:clamp}
\footnotesize
\setlength{\tabcolsep}{5pt}
\begin{tabular}{@{}lcccc@{}}
\toprule
Scale/Dataset & \texttt{pf}=32 & \texttt{pf}=64 & \texttt{pf}=128 & std \\
\midrule
2B UltraChat & 0.9967 & 0.9967 & 0.9960 & 0.0004 \\
2B LLaVA & 0.9426 & 0.9441 & 0.9562 & 0.0074 \\
2B ShareGPT4o & 0.9663 & 0.9671 & 0.9626 & 0.0024 \\
8B UltraChat & 0.9992 & 0.9993 & 0.9993 & 0.0000 \\
8B LLaVA & 0.9856 & 0.9863 & 0.9843 & 0.0010 \\
8B ShareGPT4o & 1.0000 & 1.0000 & 1.0000 & 0.0000 \\
\bottomrule
\end{tabular}
\end{table}

Because low-\texttt{pf} points are equivalent under ODB's clamp, we start the $D$ sweep in Section~\ref{sec:ablation} from $D{=}1024$ (the smallest non-clamped depth at \texttt{nw}=4, \texttt{buffer}=1024).

\section{Join Mode: Strict Per-Iteration Zero-Discard and Throughput Trade-off}
\label{app:join_tradeoff}

This appendix elaborates Theorem~\ref{thm:strict-zero}: ODB's default join-mode
termination gives a strict per-iteration $\eta_{\text{logical}}{=}0$ guarantee
(identity-level zero-discard), while non-join termination gives the
cumulative-count $\eta_{\text{quota}}{=}0$ guarantee of
Theorem~\ref{thm:zero-discard}. We quantify the throughput difference below.

\paragraph{Formal proof of Theorem~\ref{thm:strict-zero}.}
Let $M = W \cdot \lceil N/W \rceil$ be the total sampler view count and
$P = M - N$ the deterministic tail-padding overhead, so each rank starts
with quota $q = M/W = \lceil N/W \rceil$
(\texttt{DistributedSampler}, \texttt{drop\_last=False}; App.~\ref{app:proofs_state}).
Define the per-rank invariant
$\text{emitted}_r + \text{outstanding}_r + \text{remaining}_r = q$,
where $\text{outstanding}_r = |Q_r|+|B_r|$ and
$\text{remaining}_r=|R_r|$; this holds at every state transition by
Lemma~\ref{lem:no_leak} (No-Leak).
In join mode, the \texttt{done\_event} is set inside an
\texttt{all\_gather} barrier whose predicate is
$\forall r: \text{remaining}_r = 0 \wedge \text{outstanding}_r = 0$
(implemented by the drain-then-signal join predicate). Therefore at
termination $\text{emitted}_r = q$ for every rank, and
$\sum_r \text{emitted}_r = W \cdot q = M$ \emph{sampler views}.
By Lemma~\ref{lem:no_leak}, each emitted view appears in exactly one
batch, so $\eta_{\text{logical}} = 0$ over the sampler-view multiset within
a single logical iteration. The corresponding identity-level statement is
$\bigcup_r \mathrm{ids}(\text{emitted}_r) = \mathcal{I}$ (every one of the
$N$ dataset identities is emitted at least once); the $P$ surplus emits
are deterministic padding views (re-uses of shuffled-prefix identities),
not distinct dataset content. \qed

Bounded termination (Theorem~\ref{thm:deadlock-free}) is preserved
because the join-mode \texttt{all\_gather} predicate is computed from the
same broadcast tensor on all ranks (Lemma~\ref{lem:uniform_gather}); the
only difference vs.\ non-join is that the rank that first finishes its
quota waits on the same barrier instead of returning $-1$.

\paragraph{Empirical throughput cost.}
We compare default join and opt-in non-join on representative full-epoch ODB configurations from the main training stack. For each configuration, both modes use the same model, dataset, seed, $L_{\max}$, $D$, launch mode, and training hyperparameters, changing only the termination flag. Table~\ref{tab:join_tradeoff} reports literal emitted-sample throughput; the main result tables report benchmark and validation-loss metrics under default join-mode ODB.

\begin{table}[H]
\centering
\caption{Default join vs.\ opt-in non-join on representative full-epoch ODB training configurations (3 seeds, identical hyperparameters except termination flag). Throughput cells are mean$\pm$std literal emitted-sample sam/s; Join/Non and $\Delta$ are seed-paired ratios averaged before rounding the throughput columns.}
\label{tab:join_tradeoff}
\scriptsize
\setlength{\tabcolsep}{3pt}
\begin{tabular}{@{}lllrrrr@{}}
\toprule
Setting & Scale & Dataset & Default join & Opt-in non-join & Join/Non & $\Delta$ \\
\midrule
1n FFT & 2B & LLaVA      & 82.42$\pm$0.53 & 82.73$\pm$0.36 & 0.9963 & -0.37\% \\
1n FFT & 2B & ShareGPT4o & 16.09$\pm$0.21 & 16.10$\pm$0.04 & 0.9995 & -0.05\% \\
1n FFT & 2B & UltraChat  & 36.91$\pm$0.19 & 36.90$\pm$0.15 & 1.0002 & +0.02\% \\
1n FFT & 8B & LLaVA      & 24.87$\pm$0.09 & 24.91$\pm$0.06 & 0.9982 & -0.18\% \\
1n FFT & 8B & ShareGPT4o &  5.83$\pm$0.04 &  5.83$\pm$0.04 & 0.9993 & -0.07\% \\
1n FFT & 8B & UltraChat  & 10.23$\pm$0.03 & 10.23$\pm$0.03 & 1.0001 & +0.01\% \\
\midrule
1n LoRA & 2B & LLaVA      & 94.60$\pm$1.01 & 93.84$\pm$0.83 & 1.0081 & +0.81\% \\
1n LoRA & 2B & ShareGPT4o & 17.75$\pm$0.23 & 17.69$\pm$0.26 & 1.0033 & +0.33\% \\
1n LoRA & 2B & UltraChat  & 41.63$\pm$0.11 & 41.28$\pm$0.27 & 1.0084 & +0.84\% \\
1n LoRA & 8B & LLaVA      & 30.98$\pm$0.07 & 31.00$\pm$0.10 & 0.9995 & -0.05\% \\
1n LoRA & 8B & ShareGPT4o &  6.81$\pm$0.04 &  6.81$\pm$0.03 & 1.0000 & -0.00\% \\
1n LoRA & 8B & UltraChat  & 12.45$\pm$0.03 & 12.45$\pm$0.06 & 1.0002 & +0.02\% \\
\midrule
2n FFT & 2B & LLaVA      & 148.80$\pm$0.78 & 150.52$\pm$1.88 & 0.9886 & -1.14\% \\
2n FFT & 2B & ShareGPT4o &  31.35$\pm$0.37 &  31.13$\pm$0.25 & 1.0070 & +0.70\% \\
2n FFT & 2B & UltraChat  &  70.15$\pm$1.86 &  70.49$\pm$1.58 & 0.9955 & -0.45\% \\
2n FFT & 8B & LLaVA      &  45.72$\pm$0.35 &  45.75$\pm$0.35 & 0.9993 & -0.07\% \\
2n FFT & 8B & ShareGPT4o &  10.32$\pm$0.46 &  10.62$\pm$0.06 & 0.9708 & -2.92\% \\
2n FFT & 8B & UltraChat  &  18.80$\pm$0.03 &  18.35$\pm$1.23 & 1.0278 & +2.78\% \\
\midrule
2n MM-Mix & 2B & MM-Mix & 79.15$\pm$4.16 & 82.54$\pm$3.13 & 0.9612 & -3.88\% \\
\bottomrule
\end{tabular}
\end{table}

\paragraph{Interpretation.}
Across the 19 workload-level rows, the average join/non-join ratio is
0.9981 (mean $\Delta=-0.19\%$). Single-node rows range from $-0.37\%$ to
$+0.84\%$; the wider two-node/MM-Mix range is $-3.88\%$ to $+2.78\%$. Thus the
drain-before-finish barrier is not a material throughput bottleneck in these
main training settings. We therefore use join mode as the default for the
reported training rows; non-join remains an opt-in termination choice when
cumulative sample-quota closure is needed but the training stack cannot support
the join-style drain-before-finish protocol.

\paragraph{Deployment guidance.}
The reported ODB rows in the main result tables use join mode so that identity
coverage is part of the main experimental contract. The two-tier guarantee still
maps cleanly onto deployment choices:
\textbf{(default join)} for workloads where per-iteration sample composition is
part of the algorithm contract (e.g.\ curriculum schedules, RLHF rollouts,
bandit-style data selection);
\textbf{(opt-in non-join)} only for constrained training-stack integrations that
cannot support the drain-before-finish join protocol while still requiring
cumulative sample-quota closure by Theorem~\ref{thm:zero-discard} (with
empirical checks in Cor.~\ref{cor:eta_zero}). The throughput
deltas in Table~\ref{tab:join_tradeoff} are small (mean workload-level
delta $-0.19\%$, range $-3.88\%$ to $+2.78\%$ across the 19 main ODB cells), while
new workloads should re-profile the selected $D$ and $L_{\max}$ under the chosen
termination mode.

\end{document}